\documentclass[11pt]{article}
\usepackage[margin=1.5in]{geometry}

\usepackage{url}
\usepackage{booktabs}
\usepackage{calc}
\usepackage{graphicx}
\usepackage{soul}
\usepackage{multirow}
\usepackage{comment}
\usepackage{amsmath}
\usepackage{color}
\usepackage{paralist}
\usepackage{xcolor}
\usepackage{mathtools}
\usepackage{amssymb}
\usepackage{enumitem}
\usepackage{amsthm}
\usepackage{bm}
\usepackage{xspace}
\usepackage{pifont}
\usepackage{lineno}
\usepackage[utf8]{inputenc}
\usepackage{cases} 
\usepackage[caption=false, font=footnotesize]{subfig}

\usepackage[ruled,vlined,longend,linesnumbered]{algorithm2e}
\SetAlFnt{\small}
\usepackage{algorithmic}
\algsetup{linenosize=\tiny}

\usepackage{authblk}

\title{
    \titleofourpaper
}

\author[1]{Francesco Betti Sorbelli}

\author[1]{Federico Corò}

\author[2]{Sajal K. Das}

\author[1]{Cristina M. Pinotti}

\author[3]{Anil Shende}

\affil[1]{Department of Computer Science and Mathematics, University of Perugia, Italy}

\affil[2]{{Department of CS}, {Missouri University of Science and Technology}, {USA}}
\affil[3]{{Department of Math., CS, \& Physics}, {Roanoke College}, {USA}}

\newtheorem{theorem}{Theorem}

\newtheorem{lemma}{Lemma}

\newcommand*{\unit}[1]{\ensuremath{\mathrm{\,#1}}}

\newcommand{\descr}[1]{\medskip\noindent\textbf{#1}}
\renewcommand{\paragraph}{\descr}

\newcommand{\titleofourpaper}{Dispatching Point Selection for a Drone-Based Delivery System Operating in a Mixed Euclidean-Manhattan Grid\xspace}

\DeclareMathOperator*{\argmin}{arg\,min}

\newcommand{\prob}{\textsc{MEMP}\xspace}
\newcommand{\problong}{$1$-Median Euclidean-Manhattan grid Problem\xspace}

\newcommand{\copt}{\ensuremath{u^{*}}\xspace} 
\newcommand{\rowopt}{\ensuremath{r^{*}}\xspace} 
\newcommand{\colopt}{\ensuremath{c^{*}}\xspace} 

\newcommand{\hr}{\ensuremath{\overline{R}}\xspace}
\newcommand{\hc}{\ensuremath{\overline{C}}\xspace}
\newcommand{\hk}{\ensuremath{\overline{K}}\xspace}
\newcommand{\DE}[1]{\Delta_E({#1})}
\newcommand{\DM}[1]{\Delta_M({#1})}
\newcommand{\EM}{EM-grid\xspace}
\newcommand{\EMs}{EM-grids\xspace}

\newcommand{\algoptf}{OPT-F\xspace}
\newcommand{\algcmallf}{CMALL-F\xspace}
\newcommand{\algcembf}{CEMB-F\xspace}
\newcommand{\algcmebf}{CMEB-F\xspace}
\newcommand{\algbestf}{BEST-F\xspace}

\newcommand{\algoptp}{OPT-P\xspace}
\newcommand{\algcmallp}{CMALL-P\xspace}
\newcommand{\algcembp}{CEMB-P\xspace}
\newcommand{\algcmebp}{CMEB-P\xspace}
\newcommand{\algsuboptp}{S-OPT-P\xspace}

\newcommand{\cmall}{\ensuremath{u_M}\xspace}
\newcommand{\cemb}{\ensuremath{u_{\hat{C}}}\xspace}
\newcommand{\cmeb}{\ensuremath{u_{\hat{M}}}\xspace}
\newcommand{\csopt}{\ensuremath{\overline{u}}\xspace}

\newcommand{\costbar}{{\mathcal{\overline{C}}}}

\begin{document}

\maketitle

\begin{abstract}
In this paper, we present a drone-based delivery system that assumes to deal with two different mixed-areas, i.e., rural and urban.
In these mixed-areas, called \EMs, the distances are measured with two different metrics, and the shortest path between two destinations concatenates the Euclidean and Manhattan metrics.
Due to payload constraints, the drone serves a single customer at a time returning back to the dispatching point (DP) after each delivery to load a new parcel for the next customer.
In this paper, we present the \problong (\prob) for \EMs, whose goal is to determine the drone's DP position that minimizes the sum of the distances between all the locations to be served and the point itself.
We study the \prob on two different scenarios, i.e., one in which all the customers in the area need to be served (full-grid) and another one where only a subset of these must be served (partial-grid).
For the full-grid scenario we devise optimal, approximation, and heuristic algorithms, while for the partial-grid scenario we devise optimal and heuristic algorithms.
Eventually, we comprehensively evaluate our algorithms on generated synthetic and quasi-real data.
\end{abstract}

\maketitle

\section{Introduction}
Drones or Unmanned Aerial Vehicles (UAVs) are recently becoming widely used in civil applications such as environmental protection~\cite{wivou2016air,calamoneri2022realistic,justin2022integrated,qu2023environmentally}, public safety~\cite{li2015drone,he2017drone,pant2021fads}, localization~\cite{niculescu2022energy}, and smart agriculture~\cite{moribe2018combination,jawad2019wireless}.  
Currently, there is a growing interest in the use of drones in smart cities~\cite{nguyen2021drone}.
This interest is particularly increased after the global presence of the COVID-19 disease~\cite{preethika2020artificial,jat2020artificial}.
Recently, \cite{anggraeni2020deployment} discuss about the use of drones as carriers for distributing and transporting drugs and medicines, and revealed that 86.7\% of people agree that drones are more effective, surely faster, and less polluting than any ground-based distribution system.
Quarantine, closure of borders, and social distancing forced people to stay indoors for long periods, allowing them to go out only for essential activities.
Therefore, considering the need to avoid all unnecessary direct human contact, people started to rely heavily on online stores for their regular shopping.
In parallel, large companies like Amazon are testing drone-based delivery systems, particularly for what is known as ``last-mile" small item logistics.
For example, Amazon introduced ``Amazon Air Prime'', a service that uses drones able to deliver goods up to $25\unit{kg}$ to customers within a radius of $16\unit{km}$~\cite{pandit2014study,welch2015cost,shavarani2018application}, and Domino's developed a pizza-delivery service using drones~\cite{pepitone2013domino}.

There are countless advantages to using drones for deliveries, including economic benefits, saving on greenhouse gas, and the ability to deliver in time-critical situations or hard-to-reach places.
With the growth of commercial interest, researchers have begun to study variants of the Traveling Salesman Problem (TSP) for drones~\cite{ha2015heuristic}. 
However, due to the payload constraint that forces the drone to return after each delivery, TSP is not suitable for drones~\cite{sawadsitang2019multi, dayarian2020same}.
Of particular interest is the work proposed by~\cite{agatz2018optimization}, where a drone combined with a truck is used to make multiple deliveries in a given area, going back and forth from the truck. 
In that case, it is crucial to find the best location for the truck to minimize the distance traveled by the drone. 

\begin{figure}[ht]
    \centering
    \subfloat[Rural area: Miami, Florida, US.]{%
        \includegraphics[height=4.3cm]{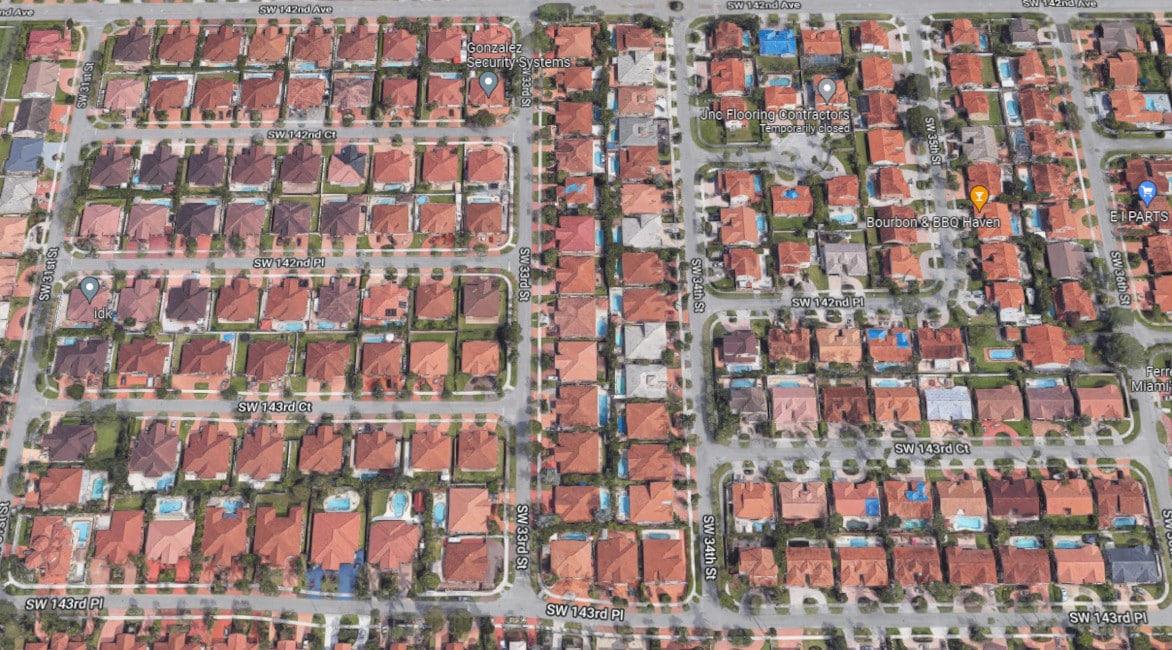}
        \label{img:rural-area}
    }
    \hfill
    \subfloat[Urban area: New York City, New York, US.]{%
        \includegraphics[height=4.3cm]{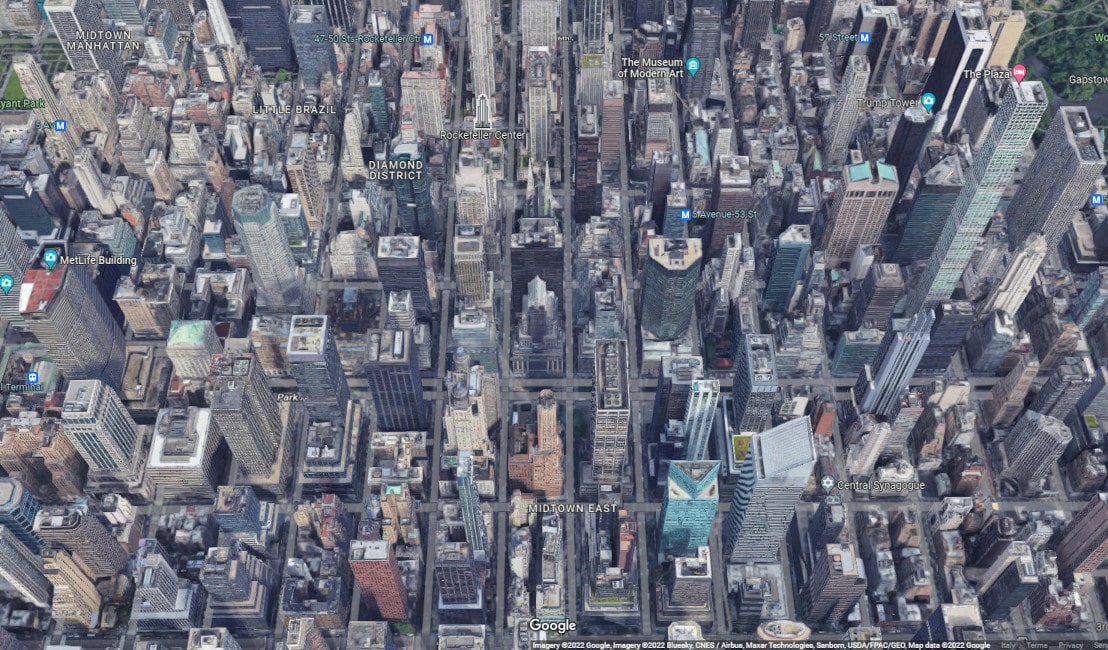}
        \label{img:urban-area}
    }
    \hfill
    \subfloat[Mixed area: Chicago, Illinois, US.]{%
        \includegraphics[height=4.3cm]{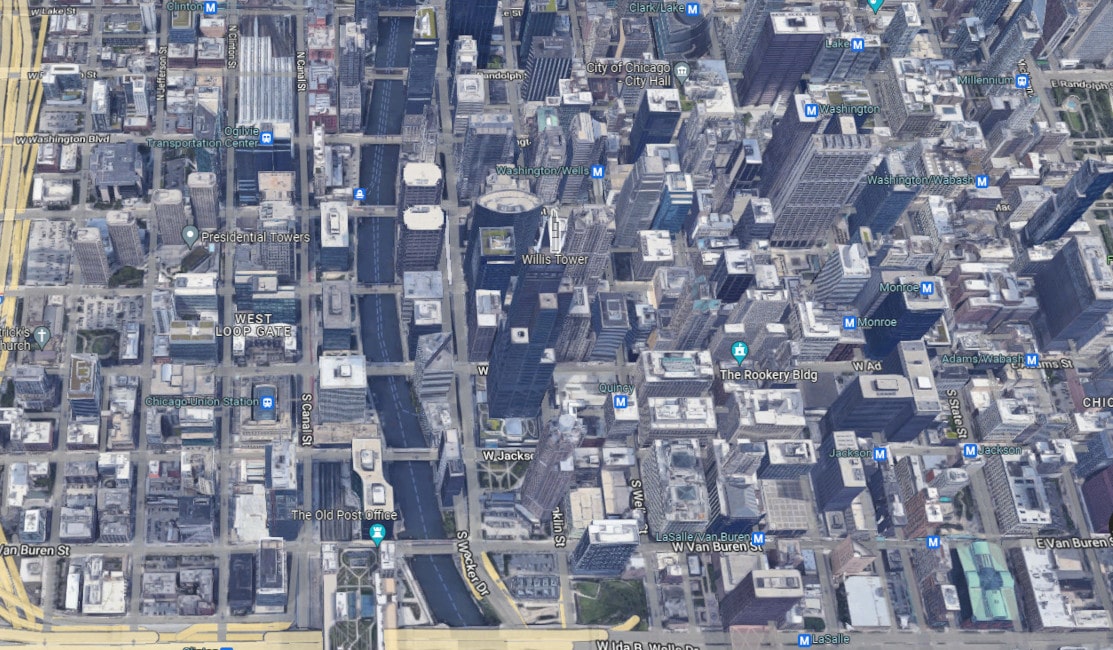}
        \label{img:mixed-area}
    }
    \caption{Examples of areas. In the rural area the drone freely moves on straight lines according to the Euclidean metric (a), while in the urban area, due to the presence of skyscrapers, the drone can only move over the streets according to the Manhattan metric (b). In the mixed area the drone has to combine both the metrics (c).}
    \label{img:areas}
\end{figure}

In this article, we imagine offering a drone delivery service to the customers of a delivery area that covers two mixed areas, i.e., a rural and/or an urban one. 
In urban areas (see Figure~\ref{img:urban-area}), for safety and privacy reasons (see regulation of inhabited centers, e.g.,~\cite{FAADrone77:online}, it is assumed that the drone flies over the streets since it cannot fly beyond a certain maximum allowed altitude.
Namely, in urban areas, it might be forbidden for the drone to travel along the straight line connecting the dispatching point (DP) and the delivery if such a straight line, due to tall obstacles, exists only at an altitude higher than the maximum altitude allowed by the regulations.
However, the drone can at least save time -- if not distance traveled -- in the urban area.
In fact, a drone is usually faster than a conventional wheeled vehicle in a crowded area: it will certainly not get stuck in traffic.
In rural areas (see Figure~\ref{img:rural-area}), instead, due to the fact that tall buildings are not present, the drone can move freely and follow the shortest path to reach its destination, making clear its advantage on the traveled distance in this case.
In the case of mixed areas (see Figure~\ref{img:mixed-area}), i.e., rural and urban, the drone must follow both the Euclidean and Manhattan metrics.
In particular, in the real example provided in Figure~\ref{img:mixed-area}, the Euclidean part resides on the left, while the Manhattan part resides on the right.
Actually, the two areas are split by the river.

We model the proposed drone delivery area as a two-dimensional mixed-grid. 
The vertices of the grid represent both the possible locations of the drone's DP and the possible delivery destinations, and they are placed in rows and columns as in a regular 2-D grid.
To model the two different types of areas through which the drone can move, we have divided the grid into two parts: the rural area and the urban area.
In the former, the distance between two destinations respects the Euclidean metric, while in the latter, the metric is Manhattan (taxicab geometry).
We are then interested in finding where to set the DP in order to minimize the sum of the distances between a subset of delivery destinations and the DP itself.
Due to strict payload constraints, the drone must return back to the DP after each delivery. 
Therefore, the drone needs to travel back and forth from the DP as many times as the number of delivery destinations on the grid. 
In this paper, we consider two delivery scenarios.
The full-grid scenario where each point of the grid has to be served by the drone.
In light of the COVID-19 pandemic, this scenario happens, for example, if the drone is used to deliver meals in a lockdown area, or to deliver self-tests to sick people.
In this particular example, the full-grid scenario is also justified by the fact that after every delivery the drone should be properly sanitized and disinfected before performing the next delivery, as proposed by~\cite{kunovjanek2021containing}.
The partial-grid scenario, instead, assumes that only a subset of points of the grid is a delivery site.
This is the usual scenario in a delivery system.

The DP selection problem in logistics, while sharing similarities with the classical facility location problem, requires a more complex modeling approach that considers factors such as varying demand and transportation costs, and it involves optimizing multiple metrics simultaneously (Euclidean and Manhattan), which is not the case for the original facility location problem.

In this paper, we present extensions and improvements to our earlier work discussed in two conference papers, i.e.,~\cite{bartoli2019exact} and~\cite{sorbelli2019automated}. 
In~\cite{bartoli2019exact} we focused on the full-grid scenario. 
In this paper we extend our work to include a partial-grid scenario. 
We suitably adapt results and ideas from~\cite{sorbelli2019automated} to this scenario, and provide new efficient algorithms for both the scenarios.
The contributions of this paper are summarized as follows.
\begin{itemize}
    \item We introduce the \EM model, which characterizes the delivery area for the drone, formed by two contiguous areas, i.e., one rural and one urban, that follow the Euclidean and Manhattan metrics, respectively.
    \item We define the \problong (\prob) and devise time-efficient algorithms for the full-grid and partial-grid scenarios.
    For the full-grid scenario we devise optimal, approximation, and heuristic algorithms, while for the partial-grid scenario we devise optimal and heuristic algorithms. 
    We also give all the proofs for the full-grid scenario that had been left blank in the previous conference paper~\cite{bartoli2019exact}.
    \item In addition to comparing the performance of our presented algorithms on randomly generated synthetic data, we also evaluate their effectiveness on quasi-real data obtained by adapting real city maps with our proposed grid model, providing a more realistic assessment of their practical utility.
    Furthermore, to ensure the accuracy of our comparison, we incorporated data from a real drone to measure the distances traveled by the dispatching point in our evaluation on quasi-real data.
\end{itemize}

The rest of the paper is organized as follows.
Section~\ref{sec:related} reviews the related work.
Section~\ref{sec:definition} formally defines \prob.
Section~\ref{sec:full} and Section~\ref{sec:partial} describe properties and algorithms for efficiently solving \prob with full-grid and partial-grid scenarios, respectively.
Section~\ref{sec:evaluation} evaluates our algorithms, and Section~\ref{sec:conclusion} offers conclusions and future research directions.

\section{Related Work}\label{sec:related}
In the literature, many works attempt to solve the drone-based last-mile delivery problem.
To the best of our knowledge, drones have been considered in a delivery system for the first time by~\cite{murray2015flying}. 
Specifically, they study the cooperation between a truck and a drone to deliver packages to customers.
The problem to solve is the Flying Sidekick Traveling Salesman Problem (FSTSP), which is a variant of the classic TSP.
In the FSTSP, a drone can autonomously perform deliveries to the customers directly flying from the main depot or can be helped by a truck.
In the latter case, the drone flies from the truck, delivers the package, and then rendezvouses with the truck again in a third location.
However, when the drone flies, the truck can do other deliveries independently, but still, it has to wait for the drone at the rendezvous location.
For solving FSTSP, the authors propose an optimal mixed-integer linear programming (MILP) and two heuristics for solving instances of practical sizes.
Then, \cite{murray2020multiple} investigate the same scenario with multiple drones in by introducing the Multiple Flying Sidekicks Traveling Salesman Problem (mFSTSP).
Even for the mFSTSP, they provide an optimal MILP formulation along with a heuristic solution approach that consists of solving a sequence of three sub-problems.

Recently, \cite{dell2022exact} present an exact formulation for FSTSP while also simplifying the model reducing the number of constraints and thus be able to solve several benchmark instances from the literature.
However, in these works, the drones fly according to the Euclidean metric.
\cite{kloster2023multiple} introduce the multiple Traveling Salesman Problem with Drone Stations (mTSP-DS), which is an extension to the classical multiple Traveling Salesman Problem (mTSP).
In this problem, multiple trucks starting to/from a single depot are in charge of supplying some packet stations that host autonomous vehicles (drones or robots).
On these stations, each truck can launch and operate drones/robots to serve customers.
The objective of mTSP-DS is to serve all customers minimizing the makespan. 
The problem is formulated as an MILP only suitable to solve small instances. 
For larger instances, many matheuristic algorithms are presented.
\cite{schermer2019hybrid} introduce the Vehicle Routing Problem with Drones and En Route Operations (VRPDERO), which is an extension to the Vehicle Routing Problem with Drones (VRPD).
In this problem, drones may not only be launched and retrieved at vertices but also on some discrete points that are located on each arc. 
The problem is formulated as an MILP, and matheuristic approaches are presented to deal with large instances.
The goal of both~\cite{schermer2019hybrid, kloster2023multiple} is to minimize the makespan, while ours is to find the best DPs from where to launch the drones.

\cite{bartoli2019exact} introduce the drone-based delivery area modeled as \EMs where a drone is used for delivering small packages to customers.
Given the delivery area divided into two contiguous areas, i.e., the rural and the urban areas, the goal is to find the optimal DP (depot or warehouse) for the drone in order to minimize the sum of all the distances between all the potential customers and the DP itself.
However, due to strict payload constraints, the drone cannot serve more than a customer at a time, and after each delivery, the drone must go back to the depot.
A similar approach, but in a different context, has been studied by~\cite{sorbelli2019automated}.
In such a scenario, the objective is to determine the optimal cart point for the drone that minimizes the distances between a set of items (on shelves) and the cart itself, assuming that shelves follow the two aforementioned metrics.
Differently from~\cite{bartoli2019exact}, in~\cite{sorbelli2019automated} only a subset of points needs to be considered.
Moreover, \cite{sorbelli2019automated} compare the current human-based system with respect to the newly proposed one based on drones.

\cite{brown2021distance} compare and contrast the performance of many algorithms in a delivery scenario.
They model the delivery area as a circular region with a central depot, while customers are randomly distributed throughout the region.
Different temporal and spatial metrics are compared when evaluating these algorithms.
In particular, they evaluate the impact of having distances measured according to both the Euclidean and Manhattan distance metrics. 
The number of customers stochastically varies under both Manhattan and Euclidean distance metrics.
The paper states that the number of customer deliveries and the metric used to measure travel distance impacts a decision maker's choice of the best algorithm and that employing multiple algorithms is recommended.

\cite{ai2021neighborhood} explore the implications and advantages of strategic planning on urban delivery services.
More specifically, the preferred method and local impacts of vehicle trips may vary by neighborhood characteristics (e.g., traffic or customer demands). 
Instead of searching for an optimal route, the paper focuses on the estimation of the vehicles' miles traveled (VMT) per meal order, considering different types of neighborhoods, delivery scenarios, and strategies. 
The proposed system is tested and evaluated in Chicago, showing that alternative delivery strategies can greatly reduce the VMT per order based on the type of neighborhood. 
In the evaluation, both the Euclidean and Manhattan metrics are combined.
However, although different metrics have been evaluated, drones have not been used in either~\cite{brown2021distance} or~\cite{ai2021neighborhood}.

Recently, there has been more effort in solving the last-mile delivery problem using drones instead of a standard vehicle, due to the flexibility of drones and their capability to fly over obstacles and avoid traffic.
\cite{poikonen2019branch} investigate the problem of solving the TSP with a Drone (TSP-D) where the drone rides on the truck. 
%
\cite{sorbelli2022scheduling} investigate the cooperation between a truck and multiple drones.
Each delivery is characterized by a drone's energy cost, a reward based on its priority, and a time interval (launch and rendezvous with the truck). 
This work aims at finding an optimal scheduling for the drones that maximizes the overall reward, subject to the drone’s battery capacity while ensuring that the same drone performs deliveries that do not overlap. 
Results show that the presented problem is $NP$-hard, therefore, different heuristics for solving the problem in a time-efficient way are proposed.
More recently, \cite{sorbelli2023wind} investigate the feasibility of performing deliveries with a drone in the presence of external factors such as wind.

\cite{li2020impact} compare the traditional truck-based delivery system against the drone-based one to reduce the general energy consumption and hence reduce the gas emissions. 
They also take into account traffic congestion.
They propose a mixed-integer green routing model with traffic restrictions and a genetic algorithm to efficiently solve the complex routing problem, showing that drones can accomplish more deliveries and at a lower cost (in terms of CO\textsubscript{2} emissions and energy consumption) compared to standard methods and that traffic types impact the results.

\cite{salama2020joint} study the last-mile scenario formed by multiple drones assisted by a single truck that carries them.
The customers to be served by the drones form a clustering, and each drone is assigned to a specific cluster.
The initial position of a drone is called a ``cluster focal point''.
Once all the focal points are computed, the truck needs to visit these points by minimizing its traveled route.
Moreover, due to payload constraints, the drones serve their customers one at a time.
The truck cannot follow the Euclidean metric, while the drones do.
The authors propose an optimal mixed integer nonlinear programming (MINLP) solution as well as an unsupervised machine learning-based heuristic algorithm.

\cite{dukkanci2021minimizing} present a variation on the theme.
The proposed model assumes that drones are assisted by trucks, that carry them through the city.
The trucks start from the main depot and park the drones in specific locations where drones have to serve the customers by a sequence of back-and-forth flights.
Both the trucks and the drones move according to the Euclidean metric.
The introduced problem is called Energy Minimizing and Range Constrained Drone Delivery Problem (ERDDP) whose objective is to minimize the total operational cost including an explicit calculation of the energy consumption of the drone as a function of the drone's speed.
The ERDDP is formulated as a second order cone program instance.

\cite{pinto2022point} propose a last-mile drone delivery scenario where multiple drones can exploit charging stations to replenish their batteries.
In this setting, the drones can fly from the main hub to the terminal station and serve, one at a time, a subset of customers in the neighborhood, and then go back to the hub.
Alternatively, they can move from a terminal station to a charging station to refill the battery and perform subsequent deliveries to other neighborhoods belonging to other terminal stations.
The objective function aims to either minimize the number of charging stations or to minimize the overall traveled distance.
The authors solve this problem by proposing an optimal and a heuristic solution.

\cite{karak2019hybrid} combine the pickup and the delivery requests in a system with stations (nodes in a given graph) and introduce the Hybrid Vehicle-Drone Routing Problem. 
Vehicles visit stations to transport delivery items and drones, while drones are launched and collected only at stations. 
The problem is formulated as a mixed-integer program, which minimizes the vehicle and drone routing cost to serve all customers.
To solve the problem, the authors use an extension of the classic Clarke and Wright algorithm (see~\cite{clarke1964scheduling}, a known heuristic to solve the Vehicle Routing Problem.
We remark that while their goal is to find the best route for drones and trucks, our goal is to find the best stations, i.e., the DPs, from where to launch the drones.

Finally, \cite{hong2018range} investigate the problem of placing drone charging facilities in an area to help increase the coverage range of drones for commercial deliveries.
The authors present an MILP formulation, and then a heuristic is proposed to effectively solve the problem.

\section{Problem Definition}\label{sec:definition}
In this section, we first introduce the delivery area model and how the drone moves inside it, and then formally describe the problem to solve.

\subsection{Delivery Area Model}
To model the delivery area, we define the \textit{Euclidean-Manhattan-Grid} (\EM) as $G=(R, C, K)$, that is a 2-D grid with $R$ rows, $C$ columns, and the \textit{Border} $B$ is the column $K \in [1, C]$ that separates the \textit{Euclidean} grid $E$ (rural area) from the \textit{Manhattan} grid $M$ (urban area) (see Figure~\ref{fig:em-grid}). 
Specifically, $E = \{1, \dots, R\} \times \{1, \dots, K\}$, $B = \{1, \ldots, R\} \times \{K\} \subseteq E$, and $M = \{1, \dots, R\} \times \{K+1, \dots, C\}$.

\begin{figure}[htbp]
    \centering
    \includegraphics[scale=1.2]{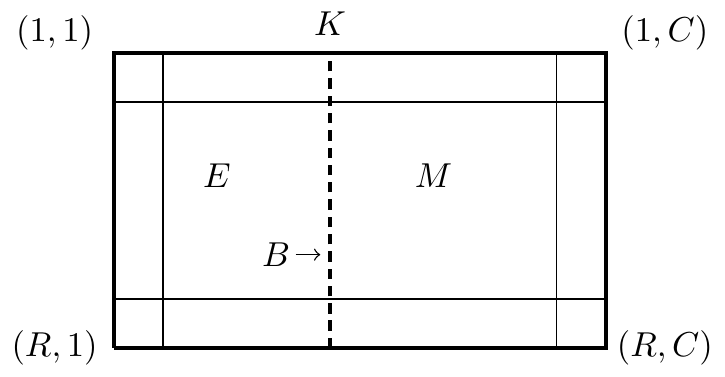}
    \caption{The \EM: $E$ models the rural area (see Figure~\ref{img:rural-area}), while $M$ models the urban area (see Figure~\ref{img:urban-area}).}
    \label{fig:em-grid}
\end{figure}

We assume that the drone delivery system covers a rectangular area: $E$ and $M$ have the same number of rows $R$.
The border consists of a single column, i.e., $K$.
Conventionally, the area consists only of a rural area region if $K=C$ (i.e., $M=\emptyset$), whereas it is effectively limited to an urban area if $K=1$.
In an \EM, there are vertices and edges connecting adjacent vertices.
Any internal vertex $u=(r_u,c_u)$ of $G$, i.e., with $1 < r_u < R$ and $1 < c_u < C$, is connected to the four adjacent vertices $(r_u,c_u \pm 1)$ and $(r_u \pm 1,c_u)$; whereas, in general, any vertex of the grid, i.e., with $1 \leq r_u \leq R$ and $1 \leq c_u \leq C$, is connected only to the existing adjacent vertices (i.e., an external vertex has only three or two neighbors).
For simplicity, we assume that the distance between any two pairs of consecutive vertices on the same row or column is constant and unitary, and so the \textit{weight} of any edge is unitary.
In this work, the ``distance'' is a measure of the \textit{time required} for performing the delivery, or of the \textit{needed energy} for shipping the package.
We also assume that every customer can be reached by the drone to/from the DP.
At the DP, the drone can recharge its battery, or just replace it with fresh ones.
Let $\hr = \lceil \frac{R}{2} \rceil$ be the middle row, let $\hc = \lceil \frac{C}{2} \rceil$ be the middle column, and let $\hk = \lceil \frac{K}{2} \rceil$ be the column that halves the Euclidean sub-grid.

For any two vertices $u$, $v$ in $G$, the distance $d(u,v)$ is the length of the shortest path traversed by the drone in the \EM to go from vertex $u$ to the destination $v$.
The Euclidean and Manhattan distances are defined, respectively, as $d_E(u, v)= \sqrt{(r_u-r_v)^2+(c_u-c_v)^2}$ and $d_M(u,v)=\vert r_u-r_v\vert +\vert c_u-c_v\vert $.

We note that the shortest path between a vertex $u \in E$ and a vertex $v \in M$ is given by $\min_{w \in B} \{d_E(u, w) + d_M(w, v)\}$.
In Lemma~\ref{lemma:path-between-set} we prove that such path is unique and passes through the vertex on the border $B$ that has the same row as $v$.

\begin{lemma}\label{lemma:path-between-set}
Consider an EM-Grid $G=(R,C, K)$. Given $u=(r_u, c_u) \in E$ and $v=(r_v, c_v) \in M$, then $d(u,v)= d_E(u, h) + d_M(h, v)$ with $h = (r_v, K)$.
\end{lemma}
\begin{proof}
Consider the vertex $h = (r_v, K)$ which shares the same row as $v$, and another vertex $w=(i, K)$ in $B$ with $h \not = w$.
We want to prove that: $d_E(u,h) + d_M(h, v) \leq d_E(u,w) + d_M(w, v)$.
This follows by the triangle inequality applied to the vertices $u, w, h$, i.e., $d_E(u,h) \leq d_E(u,w) + \vert r_v - i\vert = d_E(u,w)+d_M(w,h)$.
\end{proof}

Thus, from now on, $d(u,v)$ is given by:

\begin{equation}\label{eq:distance}
    \begin{array}{ll}
    d(u,v)=
    \begin{cases}
    d_E(u, v) & \text{if $u, v \in E$} \\
    d_M(u, v) & \text{if $u, v \in (M \cup B)$}\\
    d_E(u, h) + d_M(h, v) & \text{if $u \in E, v\in M$ where $h = (r_v, K) \in B$} \\
    d_M(u, h) + d_E(h, v) & \text{if $u\in M, v \in E$ where $h = (r_u, K) \in B$}
    \end{cases}
    \end{array}
\end{equation}

\subsection{The Column-Cost}\label{subsec:colcost}
Let the \textit{column-cost} be the distance traversed by the drone starting from a vertex $u$ in row \hr to serve all the vertices of a given column, where the column is on the same side (Euclidean or Manhattan) of the grid as the vertex $u$.
Such column-cost depends on which side of the grid the column and the vertex $u$ reside and on the number of rows $R$.
Let $\DE{j}$ be the column-cost to serve a Euclidean column at distance $j$ from the candidate DP $u=(\hr, c_u)$ with $u \in E$.
Similarly, let $\DM{j}$ be the column-cost to serve a Manhattan column at distance $j$ from the drone's DP $u=\left (\hr, c_u \right)$ with $u \in M \cup B$.
One can easily find that:
\begin{align}
    \DE{j} &= j + 2 \sum\limits_{i=1}^{\hr-1} \sqrt{i^2 + j^2} + ((R-1) \bmod 2) \sqrt{\hr^2 + j^2} \label{eq:column-cost-E} \\
    \DM{j} &= j + 2 \sum\limits_{i=1}^{\hr-1} (i + j) + ((R-1) \bmod 2) \left(\hr + j \right) \nonumber \\
    &= \hr (\hr +1 +2j)+j + ((R-1) \bmod 2) (\hr + j) \label{eq:column-cost-M}
\end{align}

\subsection{The Problem Formulation}
In our scenario, the fundamental task is to serve, with the aid of a drone, the customers of an area, e.g., to distribute viral tests to potentially infected patients.
Due to payload constraints (and, e.g., to avoid the spread of the disease), the drone cannot serve all the customers on the same flight, and it has necessarily to go back and forth from a specific position inside the delivery area (called DP, where the drone, e.g., can be also sanitized each time) to all the customers.
Specifically, this DP is the point in which all the products for customers are initially stored.
Hence, the objective is to minimize the distance traveled by the drone when it moves inside the delivery area. 
We denote this problem as the \problong (\prob) since the goal is to find a \textit{single} DP inside \EMs.

Given an \EM $G=(R, C, K)$ and a subset of vertices $H \subseteq G$, for an arbitrary vertex $u=(r_u,c_u) \in G$, we define the cost of delivery from $u$ to each point in $H$, denoted by $\mathcal{C}(H, u)$, as:
\begin{equation}\label{eq:partial}
    \mathcal{C}(H, u) = 2\sum_{v \in H} d(v, u)
\end{equation} 
As noted above, the distance between points $u$ and $v$ is a measure of the cost of delivery from $u$ to $v$, and the multiplicative constant 2 is in consideration of the round trip for each delivery.
Given $H \subseteq G$, let $H_E$ and $H_M$ be the subset of points that lie in the Euclidean and Manhattan grid, respectively, such that $H_E \cup H_M = H$ and $H_E \cap H_M = \varnothing$.
The set $H$ if formed by $n = \vert H\vert $ vertices, where $n_E = \vert H_E\vert $ and $n_M = \vert H_M\vert $, such that $n = n_E + n_M$.

When $H = G$, Eq.~\eqref{eq:partial} can be rewritten as:
\begin{equation}\label{eq:full}
    \mathcal{C}(G, u) = 2\sum_{v \in G} d(v, u) = 2 \sum_{r_v=1}^{R}  \sum_{c_v=1}^{C} d((r_v,c_v), u)
\end{equation}
with $v=(r_v, c_v) \in G$, and $r_v \in [1, R]$ and $c_v \in [1, C]$.
We refer to the scenario as a \textit{partial-grid scenario} (respectively, \textit{full-grid scenario}) when $H \subset G$ (respectively, $H = G$).
Finally, given $H \subseteq G$, we define the DP (median) \copt as:
\begin{equation}\label{eq:goal}
    \copt = \argmin_{u=(r_u,c_u) \in G} \mathcal{C}(H, u)
\end{equation}

\section{Solving \prob with Full-Grid Scenario}\label{sec:full}
In this section we give properties and then devise algorithms for solving \prob for the full-grid scenario.
This full-grid scenario is justified by the fact that a delivery company has to consider all the grid's locations as potential customers.
For example, in the case of a quarantined area, as with COVID-19, the drone can be used to deliver goods of primary necessity to all the residents in the area.
Therefore, the objective would be to find the optimal location to set the DP in order to minimize the travel distances between any customer's locations and the DP.

In the following, we first discuss how to optimally solve \prob with a full-grid (see Section~\ref{subsec:algoptf}), and then we propose an approximation algorithm, \algcmallf, that operates as if the grid is a full Manhattan grid (i.e., $K=0$) and provides a guaranteed approximation bound of $\sqrt{2}$ (see Section~\ref{subsec:algcmallf}).
Then we propose two heuristics: (1) \algcembf which assumes that all the Manhattan destinations move to the border $B$, i.e., the grid is a fully Euclidean grid (see Section~\ref{subsec:algcembf}), and (2) \algcmebf which does the opposite, i.e., the grid is a fully Manhattan grid (see Section~\ref{subsec:algcmebf}).
In Table~\ref{tab:comp_algs_full} we compare the presented algorithms that solve \prob in the full-grid scenario evaluating their time complexities and guaranteed approximation bounds.
\begin{table}[ht]
    \caption{Comparison between the algorithms that solve \prob in the full-grid scenario.}
    \label{tab:comp_algs_full}
    \centering
    \begin{tabular}{c|cccc}
        Point & Algorithm & Section & Time complexity & Approximation ratio \\
        \hline
        \copt & \algoptf & \ref{subsec:algoptf} & $\mathcal{O}(\log K)$ & $1$ \\
        \cmall & \algcmallf & \ref{subsec:algcmallf} & $\mathcal{O}(1)$  & $\sqrt{2}$ \\
        \cemb & \algcembf & \ref{subsec:algcembf} & $\mathcal{O}(1)$  & $-$ \\
        \cmeb & \algcmebf & \ref{subsec:algcmebf} & $\mathcal{O}(1)$  & $-$
    \end{tabular}
\end{table}

\subsection{Properties}\label{sec:full:properties}
In the following we prove some properties that we will exploit to devise our optimal algorithm.
We first note that with a full Manhattan grid (i.e., $K=1$) or a full Euclidean grid (i.e., $K=C$), \prob can be trivially solved.
In the former case, \prob has the Manhattan-median in $\copt=(\hr,\hc)$ (see~\cite{yamaguchi1987some}).
Note that, the median is not unique when the values $C$ and $R$ are even.
In the latter case, by using symmetry arguments, it can be proven that \prob has the Euclidean-median in $\copt=(\hr,\hc)$. 
For the general case when $1 < K < C$, we can derive properties to narrow down the set of possible median point candidates.

First, we observe that the median always belongs to the middle row $\hr$ of $G$. 
\begin{theorem}\label{thm:row-center}
Given an \EM $G=(R,C,K)$, the median $\copt=(\rowopt,\colopt)$ satisfies $\rowopt=\hr$.
\end{theorem}
\begin{proof}[Proof]
Let the \textit{row-cost} $\Gamma(u,h)$ be the distance traversed by a drone with DP $u=(x,y)$ to serve all the vertices on a row at distance $h$ from $u$. 
Note that the function $\Gamma$ depends only on the relative distance between the $x$-coordinate of $u$ and the row considered.
There are potentially two rows at distance $h$ from $u$: one above $u$ at row $x +h$, and one below $u$ at row $x -h$.
The crucial observation is that the two rows have exactly the same row cost when served by $u$.
For a fixed $u$, $\Gamma(u,h)$ increases with $h$.
In other words, for $h_2 > h_1$, it holds that $\Gamma(u,h_2) - \Gamma(u,h_1) > 0$.

Now we can show that \copt belongs to row \hr, proving that for a given $u=(\hr, j)$ and $v=(\ell, j)$, with $\ell \neq \hr$, it holds $\mathcal{C}(G, u) \leq \mathcal{C}(G, v)$.
First, we consider $\ell > \hr$.
In this case:
\begin{align}
    &\mathcal{C} (G, u) = 2 \sum_{x=1}^{\hr-1} \Gamma(u, x) + \Gamma(u, 0) \\
    &\mathcal{C} (G, v) = \sum_{x=1}^{\ell-1} \Gamma(v, x) + \sum_{x=1}^{n-\ell} \Gamma(v, x) + \Gamma(v, 0) \\
    \intertext{Subtracting $\mathcal{C} (G, u)$ from $\mathcal{C} (G, v)$, one has:}
    &\sum_{x=\hr}^{\ell-1} \Gamma(u, x) - \sum_{z=n-\ell+1}^{\hr-1} \Gamma(u, z) \ge 0 \label{eq:row-position-median} 
\end{align}
In Eq.~\eqref{eq:row-position-median} $x \ge \hr$, while $z < \hr$, and accordingly it holds $\Gamma(u, x) \ge \Gamma(u, z)$. 
Hence $\mathcal{C}(G, u) \leq \mathcal{C}(G, v)$ for any $v$.
Similarly, the result can be proven when $\ell < \hr$.
\end{proof}

Recall the notion of \textit{column-cost} defined in Section~\ref{subsec:colcost}. Recall, also, that $n = \vert G\vert $. Algebraically the following properties can be proven about the column-cost:

\begin{lemma}\label{lemma:properties-column-cost}\ 

\begin{enumerate}
    \item Both $\DE{j}$ and $\DM{j}$ increase with $j$;\label{enum:prop-1}
    \item $\DM{j} -\DM{j-t}=t\cdot n$ for $j \ge t$;
    \item $\DE{j} - \DE{j-t}$ is positive and strictly monotone increasing with $j$ for all $j > t$; \label{enum:prop-3}
    \item $\DE{j} - \DE{j-t} \ge t(\DE{j-t+1} - \DE{j-t})$, $\DE{j} - \DE{j-t} \le t(\DE{j} - \DE{j-1})$;\label{enum:prop-5}
    \item $\DE{j} < \DM{j} \le \sqrt{2}\DE{j}$.
\end{enumerate}
\end{lemma}
\begin{proof}
We prove each point separately:
\begin{enumerate}
    \item Let $j_1 > j_2$.
    From Eq.~\eqref{eq:column-cost-E}, it holds $\sqrt{i^2+j_1^2} > \sqrt{i^2+j_2^2}$.
    From Eq.~\eqref{eq:column-cost-M}, it holds $i+j_1 > i+j_2$.
    
    \item From Eq.~\eqref{eq:column-cost-M}, it holds $t + 2(\hr-1)t = tn$. 
    
    \item From Eq.~\eqref{eq:column-cost-E}, it holds $\sqrt{i^2 + j^2} > \sqrt{i^2 + (j-t)^2}$.
    Now, differentiating with respect to $j$ one obtains: 
    $${1}/{\sqrt{\left({i}/{j}\right)^2 + 1}} - {1}/{\sqrt{\left({i}/{j-t}\right)^2 + 1}}.$$
    Since $t > 0$, $j > j - t$, and hence $\frac{i}{j - t} > \frac{i}{j}$, confirming the strictly monotone increase for $j > t$.
    
    \item Observe that $\DE{j} -\DE{j-t}$ can be rewritten as the sum of the distance of consecutive columns as $\sum_{z=0}^{t-1}\DE{j-z}-\DE{j-z-1}$.
    By applying Property~\ref{enum:prop-3} with $t=1$, it holds $\DE{j} -\DE{j-t} \ge t (\DE{j-t+1} -\DE{j-t})$ because $\DE{j-z}-\DE{j-z-1}$ is increasing with $j-z$.
    Similarly, it follows: $\DE{j} -\DE{j-t} \le t( \DE{j} -\DE{j-1})$.
    
    \item Recalling the well-known Cauchy-Schwarz inequality $\sqrt{a^2+b^2} < (a+b) \le \sqrt{2} \sqrt{a^2+b^2}$, it holds: $\DE{j} < \DM{j} \le \sqrt{2}\DE{j}$.
    \popQED
\let\qed\relax
\end{enumerate}
\end{proof}

Having established in Theorem~\ref{thm:row-center} that $u^*$ is on the row $\hr$, the potential candidates for the median are vertices $(\hr, c)$. 
For a vertex that belongs to the middle row, say $(\hr, c)$, let $\costbar(c) = \mathcal{C}(G, c)$.
Selecting an arbitrary vertex $u=(\hr,c_u)$ as the candidate median, we can exploit the column-cost definition and decompose the cost $\costbar(c_u)$ into $C_1, \ldots, C_4$ as defined in Eqs.~\eqref{eq:cost-mixed_a} and~\eqref{eq:cost-mixed_b}.
Both equations coincide when $c_u=K$ because $\DM{0}=\DE{0}$ and $\DM{0}+iR=\DM{i}$. 
\begin{subnumcases}
    {\costbar(c_u)}=
    \overbrace{
    \underbrace{\sum\limits_{j=0}^{c_u-1} \DE{j} + \sum\limits_{j=1}^{K-c_u} \DE{j}}_{C_1=\text{cost($E$)}}
    +
    \underbrace{(C-K)\DE{K-c_u} + R\sum\limits_{j=1}^{C-K} j}_{C_2=\text{cost($M$)}}
    }^{\text{if}~ c_u \in E} \label{eq:cost-mixed_a}
    \\
    \overbrace{
    \underbrace{\sum\limits_{j=1}^{K} \DE{j} + R(K-1)(c_u-K)}_{C_3=\text{cost($E-B$)}}
    +
    \underbrace{\sum\limits_{j = 0}^{c_u-K} \DM{j} + \sum\limits_{j = 1}^{C-c_u}\DM{j}}_{C_4=\text{cost($M \cup B$)}}
    }^{\text{if}~ c_u \in M \cup B} \label{eq:cost-mixed_b}
\end{subnumcases}

We now prove some technical results to help to further reduce the set of median candidates.
We first show, in Lemma~\ref{lemma:reduce}, that the median cannot be ``too close'' to the left border of $G$.

\begin{lemma}\label{lemma:reduce}
The column $\colopt$ of the DP $\copt=(\hr, \colopt)$ of $G=(R,C,K)$ cannot be in the interval $[1, \ldots, \hk -1]$, where $\hk = \lceil \frac{K}{2} \rceil$ is the column that halves the Euclidean sub-grid.
\end{lemma}
\begin{proof}
This is equivalent to saying that, if $c_u \in [ 1, \hk -1 ]$, then $\costbar(c_u) > \costbar(\left\lfloor \frac{K}{2} \right\rfloor)$.

From Eq.~\eqref{eq:cost-mixed_a}, we notice that the cost $C_1$ increases if $c_u < \hk$ because $c_u$ is a sub-optimal solution for the median of the Euclidean sub-grid in $G$.
Namely, $\hk$ is the median of an \EM $G'=(R,K,K)$.
Moreover, by Lemma~\ref{lemma:properties-column-cost} (Property~\ref{enum:prop-3}), the cost $(C-K)\DE{K-c_u} > (C-K)\DE{\hk}$ because $K-c_u > \hk$.
\end{proof}

Next, we show, in Lemma~\ref{lemma:unimodal} that $\costbar(c_u)$ is convex when $c_u$ varies from $\hk$ to $K$.

\begin{lemma}\label{lemma:unimodal}
Varying $c_u$ from $\hk$ to $K$, the delivery cost function $\costbar(c_u)$ has a single minimum.
\end{lemma}
\begin{proof}
Let $t$, with $\hk \le t \le K-1$, be the vertex where the delivery cost assumes the first minimum. 
Therefore from Eq.~\eqref{eq:cost-mixed_a} we have that:
\begin{align}
\costbar(t+1) - \costbar(t) &= \DE{t+1} - \DE{K-t} \nonumber \\ 
&+(C-K)(\DE{K-(t+1)} - \DE{K-t} ) \ge 0
\end{align}

To prove that  the cost function has exactly one minimum in the interval $[\hk, \ldots, K]$, it is sufficient to show that $\costbar(t+2)-\costbar(t+1) \ge 0$ given that $\costbar(t+1)-\costbar(t) \ge 0$.

First, observe by Lemma~\ref{lemma:properties-column-cost} (Property~\ref{enum:prop-3}) that:
\begin{equation}
\DE{K-(t+2)} - \DE{K-(t+1)} \ge \DE{K-(t+1)} - \DE{K-t}
\end{equation}

Then,
\begin{align}
    \costbar(t+2)-\costbar(t+1) &= \DE{t+2} -\DE{K-(t+1)} \nonumber \\
    &+ (C-K)( \DE{K-(t+2)} - \DE{K-(t+1)} ) \nonumber \\
    &\ge \DE{t+2}-\DE{K-(t+1)} \nonumber \\
    &+ (C-K)( \DE{K-(t+1)} - \DE{K-t} ) \nonumber \\
    &= \underbrace{\DE{t+2} -\DE{t+1}}_{\ge 0} + \underbrace{\DE{K-t} - \DE{K-(t+1)}}_{\ge 0} \\
    & + \costbar(t+1)-\costbar(t) \ge \costbar(t+1)-\costbar(t) \ge 0 
\end{align}
\end{proof}

So, we know now that the optimal column cannot be in the interval $[1, \ldots, \hk-1]$, and that $\costbar(c_u)$ has a single minimum in the interval $[\hk, K]$.
We can have two cases here: $K \le \hc$ and $K > \hc$.
Lemma~\ref{lemma:unimodal-in-M} shows that when $K \le \hc$, there is only one vertex candidate as the median on the Manhattan side, while Lemma~\ref{lemma:nogreatercm} establishes the fact that when $K > \hc$, the median cannot be in $[\hc, C]$.

\begin{lemma}\label{lemma:unimodal-in-M}
Let $K \le \hc$. 
Varying $c_u$ in the Manhattan side, $K \le c_u \le C$, the delivery cost function $\costbar(c_u)$ has a single minimum in $\hc$.
\end{lemma}
\begin{proof}
Since the candidate median is on the Manhattan side, the delivery cost is expressed by Eq.~\eqref{eq:cost-mixed_b}.
Moving the candidate from $(\hr,c_u)$ to $(\hr,c_u+1)$, the vertices on the left of $c_u$ (including $c_u$ itself) increase by exactly one their distance from the median.
Whereas, the vertices on the right of the column $c_u$ decrease their distance from the median by one.
Hence, the delivery cost function decreases as long as $K \le c_u \le \hc$. 
\end{proof}

\begin{lemma}\label{lemma:nogreatercm}
Let $K>\hc$.
The value of the median candidate column cannot be greater than $\hc$.
\end{lemma}
\begin{proof}
We can first exclude any candidate with $c_u \ge K$ because, as we seen in Lemma~\ref{lemma:unimodal-in-M}, the delivery cost function $\costbar(c_u)$ is increasing for $c_u \ge K$.
Moreover, in order to exclude the candidates in the range $[\hc, K]$, we first prove that $\costbar(\hc+1)>\costbar(\hc)$. 
Namely,
\begin{align}
\costbar(\hc + 1) - \costbar(\hc) &= \DE{\hc + 1} - \DE{K - \hc} \nonumber \\
&- (C-1-K )\left (\DE{K-\hc} - \DE{K-(\hc+1)}\right ) \ge 0
\end{align}

because, by Property~\ref{enum:prop-5}) and~\ref{enum:prop-3}) of Lemma~\ref{lemma:properties-column-cost}
\begin{align}
\DE{\hc+1}-\DE{K-\hc} &> \left (2\hc+1-K \right )\left(\DE{K-\hc+1}-\DE{K-\hc} \right ) \nonumber \\
&>(2\hc+1-K) \left(\DE{K-\hc}-\DE{K-\hc-1} \right ) \nonumber \\ 
&\ge(C-1-K)\left (\DE{K-\hc}-\DE{K-(\hc+1)} \right )
\end{align}

Since we have proven in Lemma~\ref{lemma:unimodal} that if the cost function is increasing in one vertex of row \hr, then it is increasing in all the vertices on its right, there are no candidates for the median greater than \hc.
\end{proof}

Then, Theorem~\ref{thm:optf-correctness} follows immediately from Theorem~\ref{thm:row-center}, and Lemmas~\ref{lemma:reduce},~\ref{lemma:unimodal-in-M}, and~\ref{lemma:nogreatercm}.
\begin{theorem}\label{thm:optf-correctness}
Suppose $\copt = (r^*, c^*)$ be the median for $G=(R,C,K)$. Let $c, \hk \leq c \leq K$ be such that
\[
\costbar(c) = \min_{i \in [\hk, K]} \{\costbar(i)\}
\]
and, let $c', \hk \leq c' \leq \hc$ be such that
\[
\costbar(c') = \min_{i \in [\hk, \hc]} \{\costbar(i)\}
\]
Then, $r^* = \hr$, and
\[
c^* = \begin{cases}
        c & \mbox{if } K \leq \hc$ \mbox{ and } $\costbar(c) < \costbar(\hc)\\
        \hc &  \mbox{if } K \leq \hc$ \mbox{ and } $\costbar(c) \geq \costbar(\hc)\\
        c' & otherwise
      \end{cases}
      \]
\qed
\end{theorem}

\subsection{The Optimal Algorithm \algoptf}\label{subsec:algoptf}
Theorem~\ref{thm:optf-correctness} directly translates to algorithm \algoptf (see the pseudo-code in Algorithm~\ref{alg:algoptf}) that optimally solves \prob in the full-grid scenario.

\begin{algorithm}[ht]
    \caption{The \algoptf Algorithm}
    \label{alg:algoptf}
    \begin{algorithmic}[1]
        \IF {$K \le \hc$} \label{code:algoptf:case1}
            \STATE $c \gets \texttt{find-minimum}(\hk, K)$\;\label{code:algoptf:init-case1}
            \IF {$\costbar(c) < \costbar(\hc)$}
                \STATE \RETURN{$(\hr, c)$} \COMMENT{see Lemma~\ref{lemma:unimodal}}
            \ELSE
                \STATE \RETURN{$(\hr, \hc)$} \COMMENT{see Lemma~\ref{lemma:unimodal-in-M}}
            \ENDIF
        \ELSE \label{code:algoptf:case2}
            \STATE $c \gets \texttt{find-minimum}(\hk, \hc)$\label{code:algoptf:init-case2} \COMMENT{see Lemma~\ref{lemma:nogreatercm}}
            \STATE \RETURN{$(\hr, c)$}
        \ENDIF
    \end{algorithmic}    
\end{algorithm}

Given a closed interval of column numbers, the procedure \texttt{find-minimum} returns the column number, $c$, in that interval such that the total delivery cost from the DP $(\hr, c)$ is the least over all the vertices $(\hr, i)$ for $i$ in the given closed interval.

\subsubsection{Time Complexity of Algorithm \algoptf}\label{subsubsec:algoptf-time}

Note that the minimum in $[\hk, K]$, returned by invoking the procedure \texttt{find-minimum}, can be found by applying a binary search due to the unimodality proven in Lemma~\ref{lemma:unimodal}.
Similarly, when $K > \hc$ (Line~\ref{code:algoptf:case2}), the minimum is in the interval $[\hk, \hc]$, i.e., a sub-interval of $[\hk, K]$. Thus, as above, it can be found through the \texttt{find-minimum} procedure in Line~\ref{code:algoptf:init-case2}.
The time complexity of the \texttt{find-minimum} procedure is logarithmic in the width of the sub-interval where the minimum resides.
The interval has a width of $\frac{K}{2}$ in Line~\ref{code:algoptf:init-case1}, and a width of $\frac{C}{2} - \frac{K}{2} \le K - \frac{K}{2} = \frac{K}{2}$ in Line~\ref{code:algoptf:init-case2} since $K \ge \frac{C}{2}$. 
Thus, we can conclude that, in each case, the time complexity of \texttt{find-minimum} is $\mathcal{O}(\log K)$. 

With regards to the time complexity of the \algoptf algorithm, we observe that for a fixed $c_u$, the delivery cost $\costbar(c_u)$ can be computed in $\mathcal{O}(1)$ time by applying Eq.~\eqref{eq:cost-mixed_a} if the prefix sums of the column-cost are computed in a pre-processing phase.
The column-costs $\DE{j}$ and their prefix sums $\sum_{t=1}^{j} \DE{t}$, for $1 \le j \le K$, can be computed and memorized in a vector in $\mathcal{O}(RK+K)$ time.
For $1 \le j \le K$, the $K$ column-costs $\DM{j}$ and their prefix sums $\sum_{t=1}^{j} \DM{t}$ can be computed and memorized in a vector in $\mathcal{O}(K)$ time by using the closed-form in Eq.~\eqref{eq:column-cost-M}.
Then, assuming that the prefix-sums of the column-costs are given as input to the algorithm, i.e., they are computed in a pre-processing phase, each $\costbar(c_u)$ can be computed in constant time. 
Thus, the optimal point \copt can be computed by the \algoptf algorithm in $\mathcal{O}(\log K)$ time.

\subsection{Approximation and Heuristic Algorithms}\label{subsec:approx-full-grid}

In the previous section, we presented an algorithm that optimally solves \prob in the full-grid scenario, taking logarithmic time in the number of columns of the grid. In this section, we present three constant-time algorithms, \algcmallf, \algcembf, and \algcmebf, for solving \prob in the full-grid scenario. 
We also establish an upper bound on the approximation ratio for algorithm \algcmallf. 
In Section~\ref{subsubsec:full-synthetic}, we present comparative empirical performance evaluation of these three approximation algorithms using synthetic data.

\subsubsection{Algorithm \algcmallf}\label{subsec:algcmallf}

Algorithm \algcmallf returns the Manhattan-median $\cmall = (\hr,\hc)$ as the DP. 
That is, \algcmallf operates as if the grid is a full Manhattan grid (i.e., $K=1$)
This algorithm is sub-optimal for $1 < K < C$ providing a guaranteed approximation bound of $\sqrt{2}$, while it is optimal when $K=1$ or $K=C$.

\begin{lemma}\label{lemma:algcmallf}
The \algcmallf algorithm provides a $\sqrt{2}$ approximation ratio when $1 < K < C$.
\end{lemma}
\begin{proof}
Let $\copt$ be the median for $G=(R,C,K)$.
Since the Euclidean column-cost $\DE{j}$ is smaller than the Manhattan column-cost $\DM{j}$, $\mathcal{C}(G,\cmall) \le \mathcal{C}_M(G,\cmall)$, where $\mathcal{C}_M(G,\cmall)$ is the cost when $K=1$.
Moreover, $\mathcal{C}(G,\copt) > \mathcal{C}_E(G,\copt)$ because the Manhattan distance is at least as much as the Euclidean distance. Then, $\mathcal{C}_E(G,\copt) > \mathcal{C}_E(G,\cmall)$ because \cmall is  the Euclidean-median of $G=(R,C,C)$.
Thus, by the Cauchy-Schwarz inequality in Lemma~\ref{lemma:properties-column-cost}: 
\[
\frac{\mathcal{C}(G,\cmall)}{\mathcal{C}(G,\copt)} < \frac{\mathcal{C}_M(G,\cmall)}{\mathcal{C}_E(G,\cmall)} \le \sqrt{2}.
\]
\let\qed\relax
\end{proof}

The \algcmallf algorithm finds the point \cmall in constant time, so its time complexity is $\mathcal{O}(1)$.

\subsubsection{Algorithm \algcembf}\label{subsec:algcembf}
Algorithm \algcembf solves \prob with full-grid, and selects the DP $\cemb=(\hr, \mu)$ with 
\begin{equation}
    \mu=\frac{(\sum_{i=1}^{K}i)+(C-K)K}{C}=\frac{K(2C-K+1)}{2C}
\end{equation}
\algcembf imagines that all the Manhattan destinations move on the border $B$.
Thus, the grid becomes a full Euclidean grid $E'$ with $K$ columns, whose rightmost column has multiplicity $w_v=C-K$.
So, minimizing Eq.~\eqref{eq:cost-mixed_a} is the same as finding the Euclidean-median of $E'$.
Unfortunately, there is no closed form to compute the exact Euclidean-median for a set of points\footnote{We only know the Euclidean-median of a Euclidean grid $G=(R,C,C)$, whose columns have all multiplicity $1$.}.

Algorithm \algcembf finds the point \cemb using a constant number of operations. 

\subsubsection{Algorithm \algcmebf}\label{subsec:algcmebf}

Algorithm \algcmebf imagines that all the Euclidean destinations are moved on the border $B$.
Thus, the grid becomes a full Manhattan grid $M'$ with $C-K$ columns, whose leftmost column has multiplicity $w_v=K$.
Therefore, minimizing Eq.~\eqref{eq:cost-mixed_b} is the same as finding the Manhattan-median of the grid $M'$.
Algorithm \algcmebf solves \prob with full-grid, and selects the DP $\cmeb=(\hr, \mu)$ with $\mu = K$ if $K \ge \hc$ or with $\mu = \hc$ if $K < \hc$. In either case, $\cmeb$ is computed in constant time.

\section{Solving \prob with Partial-Grid Scenario}\label{sec:partial}

In this section, we focus on \prob in the partial-grid scenario. In this case, we are given $H \subset G$ as the $n$ delivery points, i.e., $\vert H\vert  = n$.

We first discuss the trivial cases where $K = 1$, i.e., a full Manhattan grid, and $K = C$, i.e., a full Euclidean grid (see Section~\ref{subsec:partial-trivial}). The rest of the sections deals with the general cases. 

In Sections~\ref{subsec:algcembp} and~\ref{subsec:algcmebp} we propose two heuristic algorithms, \algcembp and \algcmebp, assuming the optimal DP is in the Euclidean side of the grid, and in the Manhattan side of the grid, respectively. Note that neither of these two algorithms returns the optimal DP. Algorithm \algcembp returns the best DP candidate in $E$, whereas \algcmebp returns the best DP candidate in $M$.

Then, in Section~\ref{subsec:algoptp} we combine the above two heuristic algorithms and present algorithm \algoptp\ to find the optimal median, the DP, in $G$.

In Sections~\ref{subsec:algsuboptp} and~\ref{subsec:algcmallp} we present two algorithms, \algsuboptp\ and \algcmallp, respectively, that each compute a sub-optimal DP, but are more efficient than Algorithm \algoptp.

To summarize the above, we tabulate in Table~\ref{tab:comp_algs_partial} the presented algorithms that solve \prob in the partial-grid scenario along with their time complexities and guaranteed approximation bounds. 
In Section~\ref{subsubsec:partial-synthetic}, we present comparative empirical performance evaluation of these algorithms using synthetic data, while in Section~\ref{sec:real}, we compare these algorithms using quasi-real data.

\begin{table}[ht]
    \caption{Comparison between the algorithms that solve \prob in the partial-grid scenario.}
    \label{tab:comp_algs_partial}
    \centering
    \begin{tabular}{c|cccc}
        Point & Algorithm & Section & Time complexity & Approximation ratio \\
        \hline
        \cemb & \algcembp & \ref{subsec:algcembp} & $\mathcal{O}(nR \log K)$  & $-$ \\
        \cmeb & \algcmebp & \ref{subsec:algcmebp} & $\mathcal{O}(nR)$  & $-$ \\
        \copt & \algoptp & \ref{subsec:algoptp} & $\mathcal{O}(nR \log K)$ & $1$ \\
        \csopt & \algsuboptp & \ref{subsec:algsuboptp} & $\mathcal{O}(n \log R \log K)$  & $-$ \\
        \cmall & \algcmallp & \ref{subsec:algcmallp} & $\mathcal{O}(n)$  & $-$
    \end{tabular}
\end{table}

\subsection{$K=1$ or $K=C$}\label{subsec:partial-trivial}

We first discuss how to optimally solve \prob to serve a subset of customers on a grid that is not mixed, i.e., either Manhattan or Euclidean. 
In a Manhattan grid (i.e., an \EM with $K=1$), given a subset $H \subset G$ of customers, \prob with partial-grid scenario can be trivially solved in $\mathcal{O}(\vert H\vert )$ time by selecting $\copt=(r_{u^*},c_{u^*})$ where $r_{u^*}$ is the median of the row coordinate of the customers and $c_{u^*}$ is the median of the column coordinate of the customers (see~\cite{yamaguchi1987some}). 

In the literature, there are many results concerning the problem, called \textit{geometric median}, of determining the DP that minimizes the sum of distances between the points of a given set $H \subset \mathbb{R}^d$  and the DP itself.
The geometric median is denoted as the Euclidean-median in two-dimensional space. 
Although the Euclidean-median is unique and the sum $\mathcal{C}(H, u)$ of the distances from each customer to the DP $u$ is positive and strictly convex in $\mathbb{R}^d$, as proven by~\cite{vardi2000multivariate}, there is no exact and closed expression for the Euclidean-median of an arbitrary set $H$ of real points. 
In a Euclidean grid (i.e., an \EM with $K = C$), \prob with partial-grid scenario can be solved more easily because the candidate positions in the plane are just the vertices of $G$, and the number of attempts to determine the single Euclidean-median is limited by the fact that \EM is formed by $R$ rows and $C$ columns. 
So, a trivial solution for \prob with partial-grid takes $\mathcal{O}(RC\vert H\vert )$ time because for each position $u$ of the grid the cost $\mathcal{C}(H, u)$ can be computed in time $\mathcal{O}(\vert H\vert )$. 
Exploiting the fact that, in $\mathbb{R}^2$, $\mathcal{C}(H, u)$ is positive and strictly convex  when $u$ moves on a single row of the grid, the position in a row that provides the minimum cost $\mathcal{C}(H, u)$ can be computed by a slightly modified binary search. 
So, \prob with a subset of customers on a Euclidean grid can be solved in $\mathcal{O}(R \log C \vert H\vert )$ time because for each row of the grid only $\log C$ candidates are tested.

In the next section, for the general case, i.e., the mixed-grid with $1 < K < C$, we leverage the observation that the median \copt resides either in $M$ or in $E$.

\subsection{The \algcembp Algorithm}\label{subsec:algcembp}

As noted above, \algcembp solves \prob assuming that $u^* \in E$ and returns $\cemb \in E$. Algorithm \algcembp selects as the DP the point in the Euclidean side $E$ that returns the minimum cost.

For any DP $u \in E$, by Lemma~\ref{lemma:path-between-set}, for each point $v$ in $H_M$ (the subset of $H$ that contains the vertices in the Manhattan side), the drone must travel horizontally from the border $B$ to $v$.
So, the drone covers the same fixed cost in the Manhattan grid regardless of the position of $u \in E$. This motivates our algorithm whose pseudo-code is presented in Algorithm~\ref{alg:algcembp}. 

Similar to \algcembf for the full-grid scenario, we construct the multi-set $H_{M}'$ consisting of the projections of all the points in $H_M$ to the border $B$ (Algorithm~\ref{alg:algcembp}, Line~\ref{code:algcembp:init}). Then, we compute the point $\cemb \in E$ that minimizes the sum of the distances to the points in the multi-set $H' = H_E \cup H_{M}'$, taking into account the multiplicity of points in $H'$. 
Towards this, for each row (Algorithm~\ref{alg:algcembp}, Line~\ref{code:algcembp:row}) we evaluate the minimum from the first to the $K^{th}$ column, and update the overall minimum, if necessary.
By exploiting the fact that in the Euclidean space the function $\mathcal{C}(H, u)$ to minimize is positive and strictly convex (see~\cite{vardi2000multivariate}), we know that $\mathcal{C}(H, u)$ is unimodal when fixing a row and varying the columns.
Hence, we can calculate the minimum on each row by performing a time-efficient binary search that requires logarithmic time (Algorithm~\ref{alg:algcembp}, Line~\ref{code:algcembp:ternary}).
Since the cost of movement in $M$ does not change the minimum, the returned value, $\cemb$, is the best DP in $E$.

\begin{algorithm}[ht]
    \caption{The \algcembp Algorithm}
    \label{alg:algcembp}
    \begin{algorithmic}[1]
        \STATE $H_{M}' \gets \{(r_u, K) \in B \mid  u \in H_M\}, H' \gets H_E \cup H_{M}'$\;\label{code:algcembp:init}
        \STATE $\cemb \gets \varnothing, cost \gets + \infty$\;
        \FOR {$i \in 1, \ldots, R$} \label{code:algcembp:row}
            \STATE $\copt_i \gets \texttt{find-minimum-on-row}(i, H', 1, K)$\;\label{code:algcembp:ternary}
            \IF {$\mathcal{C}(H', \copt_i) < cost$} \label{code:algcembp:if}
                \STATE $\cemb \gets \copt_i, cost \gets \mathcal{C}(H', \copt_i)$\;
            \ENDIF
        \ENDFOR
        \STATE \RETURN{$\cemb$}
    \end{algorithmic} 
\end{algorithm}

About the time complexity, due to the fact that there are $n$ points to serve, and since we perform $R$ binary searches (one for each row of a Euclidean grid with $K$ columns), the total cost of \algcembp is $\mathcal{O}(n R \log K)$.

\subsection{The \algcmebp Algorithm}\label{subsec:algcmebp}

Algorithm \algcmebp solves \prob assuming that $\copt \in M$ and returns $\cmeb \in M$.

For any DP, $u$ in $M$, by Lemma~\ref{lemma:path-between-set}, 
the drone has to fly through the projection of $u \in B$ to serve any point of $H_E$. Thus, for any DP $u$, there is one intermediate point in $B$ for all drone paths from $u$ to points in $H_E$. Note that the intermediate point does not depend on the column $c_u$ of the DP $u$. This motivates our algorithm whose pseudo-code is presented in Algorithm~\ref{alg:algcmebp}.

We first compute the column median, $\chi$, of the points in $H'$, consisting of the points in $H_M$ and the points in $H_E$ moved to the border. The function \texttt{column-median} takes into account the possible multiplicity of points in $H'$. Since we are only concerned with the column numbers, we can move each point in $H_E$ to any point on the border; we move all of them to $(1, K)$  (Algorithm~\ref{alg:algcmebp}, Lines~\ref{code:algcmebp:H1} and~\ref{code:algcmebp:median}).

Since the median in $M$ can reside in $R$ different rows, we have $R$ candidate intermediate points, which are the points of the border $B$. For each row $i$, we
construct the multi-set $H_{E}'(i)$ consisting of all the points in $H_E$ moved to the intermediate point $(i, K)$ on the border $B$ (Algorithm~\ref{alg:algcmebp}, Line~\ref{code:algcmebp:project}). 
We then consider the point $u_i^* = (i, \chi) \in M$ as the median of  the points in the multi-set $H' = H_M \cup H_{E}'(i)$.
With $u_i^*$ as the DP, the cost of delivery is the sum of two costs: (1) the cost of flying the drone between $u_i^*$ and each point in $H'$, with all distances calculated according to the Manhattan metric, and (2) the cost of flying the drone between each point in $H_E$ and the intermediate point corresponding to $u_i^*$, with all distances calculated according to the Euclidean metric (Algorithm~\ref{alg:algcmebp}, Line~\ref{code:algcmebp:if}). 

The algorithm returns as $\cmeb$ the intermediate point that witnesses the least cost, over all the possible intermediate points, and is thus the best DP in $M$.

\begin{algorithm}[ht]
    \caption{The \algcmebp Algorithm}
    \label{alg:algcmebp}
    \begin{algorithmic}[1]
        \STATE $\cmeb \gets \varnothing, c \gets + \infty$\;
        \STATE $H_{E}'(1) \gets \{(r_u, 1) \in B \mid u \in H_E\}$, $H' \gets H_{E}'(1) \cup H_M$\;\label{code:algcmebp:H1}
        \STATE $\chi \gets \texttt{column-median} (H')$\;\label{code:algcmebp:median}
        \FOR {$i \in 1, \ldots, R$} \label{code:algcmebp:row}
            \STATE $H_{E}'(i) \gets \{(r_u, i) \in B \mid u \in H_E\}$, $H' \gets H_{E}'(i) \cup H_M$\;\label{code:algcmebp:project}
            \STATE $\copt_i \gets (i, \chi)$\;
            \STATE $C_i \gets \mathcal{C}(H', \copt_i) + cost(H_E \rightarrow (i,K))$\;
            \IF {$C_i < c$} \label{code:algcmebp:if}
                \STATE $\cmeb \gets \copt_i$\;
                \STATE $c \gets C_i$\;
            \ENDIF
        \ENDFOR
        \STATE \RETURN{$\cmeb$}
    \end{algorithmic} 
\end{algorithm}

As for the complexity of the algorithm, we compute the column median of $n$ points once. Then, for every row $i$, the cost $cost(H_E \rightarrow (i,K))$ has to be computed, and this requires $\mathcal{O}(\vert H_E \vert)$ time.
Thus, the algorithm's complexity is $\mathcal{O} (n + R\vert H_E\vert )$.

\subsection{The Optimal \algoptp Algorithm}\label{subsec:algoptp}

Having computed the best DPs on both sides of \EM we can now optimally solve \prob with partial-grid scenario.

The \algoptp algorithm finds the optimal point \copt by comparing the best point between \cemb and \cmeb.
Given that Algorithm \algcembp returns the best DP in $E$, and Algorithm \algcmebp returns the best DP in $M$, the simple idea of \algoptp is to compare these two points and return the best one, as follows:
\begin{equation}
    \copt = \argmin \{ \mathcal{C}(H,\cemb), \mathcal{C}(H,\cmeb) \}.
\end{equation}
About the time complexity, recalling that \algcembp takes $\mathcal{O}(n R \log K)$ and \algcmebp takes $\mathcal{O}(n R)$, the overall time complexity of \algoptp is $\mathcal{O}(n R \log K)$.

\subsection{The \algsuboptp Algorithm}\label{subsec:algsuboptp}

Algorithm \algsuboptp applies the same strategy as \algoptp, i.e., of comparing the best among two points: the best one in the Euclidean grid, and the best one in the Manhattan grid. Nonetheless, it uses slightly different versions of both \algcembp and \algcmebp. 

In the modified version of Algorithm \algcembp, we perform two binary searches -- one on the rows plus the one on the columns up to the $K^{th}$ column. In the modified version of Algorithm \algcmebp, we perform a single binary search on the rows. Thus, the time complexities of the two algorithms are $\mathcal{O}(n R \log K)$ and $\mathcal{O}(n R)$, respectively.

Thus, the overall time complexity of \algsuboptp is $\mathcal{O}(n \log R \log K)$. However, this strategy does not guarantee that the returned point \csopt is optimal.

\subsection{The \algcmallp Algorithm}\label{subsec:algcmallp}

Essentially, Algorithm \algcmallp ignores the border that separates the Euclidean and the Manhattan sides, and computes the DP \cmall as if $G=(R,C,1)$. The algorithm returns
\begin{equation}\label{eq:cmall}
    \cmall = ( \tilde{r}_H, \tilde{c}_H ),
\end{equation}
where $\tilde{r}_H$ and $\tilde{c}_H$ are the individual medians of the rows and the columns, respectively, of the points in $H$ (see~\cite{yamaguchi1987some}).
Note that, also in this scenario, the median is not unique if $\vert H\vert $ is even.

Although \algcmallp optimally solves \prob in the case of $G = (R, C, 1)$, it is sub-optimal in the general case. 
Nonetheless, it works in linear time (see~\cite{cormen2009introduction}) with respect to the number of points.

\section{Performance Evaluation}\label{sec:evaluation}
In this section, we empirically compare the performance of our algorithms in terms of the quality of the solution (i.e., delivery cost), and their running times, for solving \prob in both scenarios.

\subsection{Settings and Parameters}
We implemented our algorithms in Python language version 3.9, and run all the instances on an Intel i7-860 computer with $12 \unit{GB}$ of RAM.
In order to evaluate our proposed algorithms for solving \prob in both scenarios, we rely on synthetic and quasi-real delivery areas.

For the \textit{synthetic case} (Section~\ref{sec:synthetic}),
we set different layouts by varying $R, C \in \{50, \ldots, 400\}$ and $1 \le K \le C$.
Then, we compare the algorithms with respect to the optimal one, and we plot the experimental ratio $\rho = \frac{\mathcal{C}(H,\tilde{u})}{\mathcal{C}(H,\copt)} \ge 1$.
In other words, for $H \subseteq G$, $\rho$ is the ratio between the total cost for serving all the required customers from the DP  $\tilde{u}$ as returned by the tested algorithm, and the total cost of the optimal solution where the DP used is \copt.
When testing the full-grid scenario, we compare \algcmallf, \algcembf, and \algcmebf with respect to the optimal algorithm \algoptf, while when testing the partial-grid scenario, we compare \algcmallp, \algcembp, \algcmebp, and \algsuboptp with respect to the optimal algorithm \algoptp.

Moreover, in the partial-grid scenario, we uniformly generate $n=\vert H\vert $ random positions inside the grid with $n=\{5, \ldots, 100\}$, and then return the \textit{average} ratio (along with the \textit{standard deviation}) on $33$ random instances.
Also in the partial-grid scenario, given a setting with $n$ random customers, we evaluate the algorithms when balancing the quantities $n_E$ and $n_M$ with respect to a certain fraction $p=\{0, \frac{1}{3}, \frac{1}{2}, \frac{2}{3}, 1\}$ on $n$, such that $n_E = p \cdot n$ and $n_M = (1-p) \cdot n$, with $n = n_E + n_M$.

For the \textit{quasi-real case} (Section~\ref{sec:real}), we only test the more general partial-grid scenario in random instances taken from real cities, like the ones shown in Figure~\ref{img:areas}.
For these examples, we approximately extract the actual \EM from the map, and then we run our proposed algorithms.
Obviously, real cities cannot be exactly modeled as \EMs due to the fact that roads and buildings can be arbitrarily made by people.
However, we have found interesting examples and attempted to perform our algorithms on these layouts.
Clearly, in the aforementioned grids of customers, some houses or skyscrapers can be missing.

\subsection{Results with Synthetic Data}\label{sec:synthetic}

\subsubsection{Full-Grid Scenario}\label{subsubsec:full-synthetic}

We first analyze our empirical results of the performance of the algorithms with respect to the delivery costs, and then with respect to the running times.

\medskip\noindent\textit{Delivery Costs}

Figure~\ref{fig:full-three} compares the algorithms when solving \prob with the full-grid scenario reporting, for each plot, the ratio $\rho$ between the total cost of the tested algorithm and the optimal total cost.

\begin{figure}[ht]
    \centering
    \includegraphics[scale=0.9]{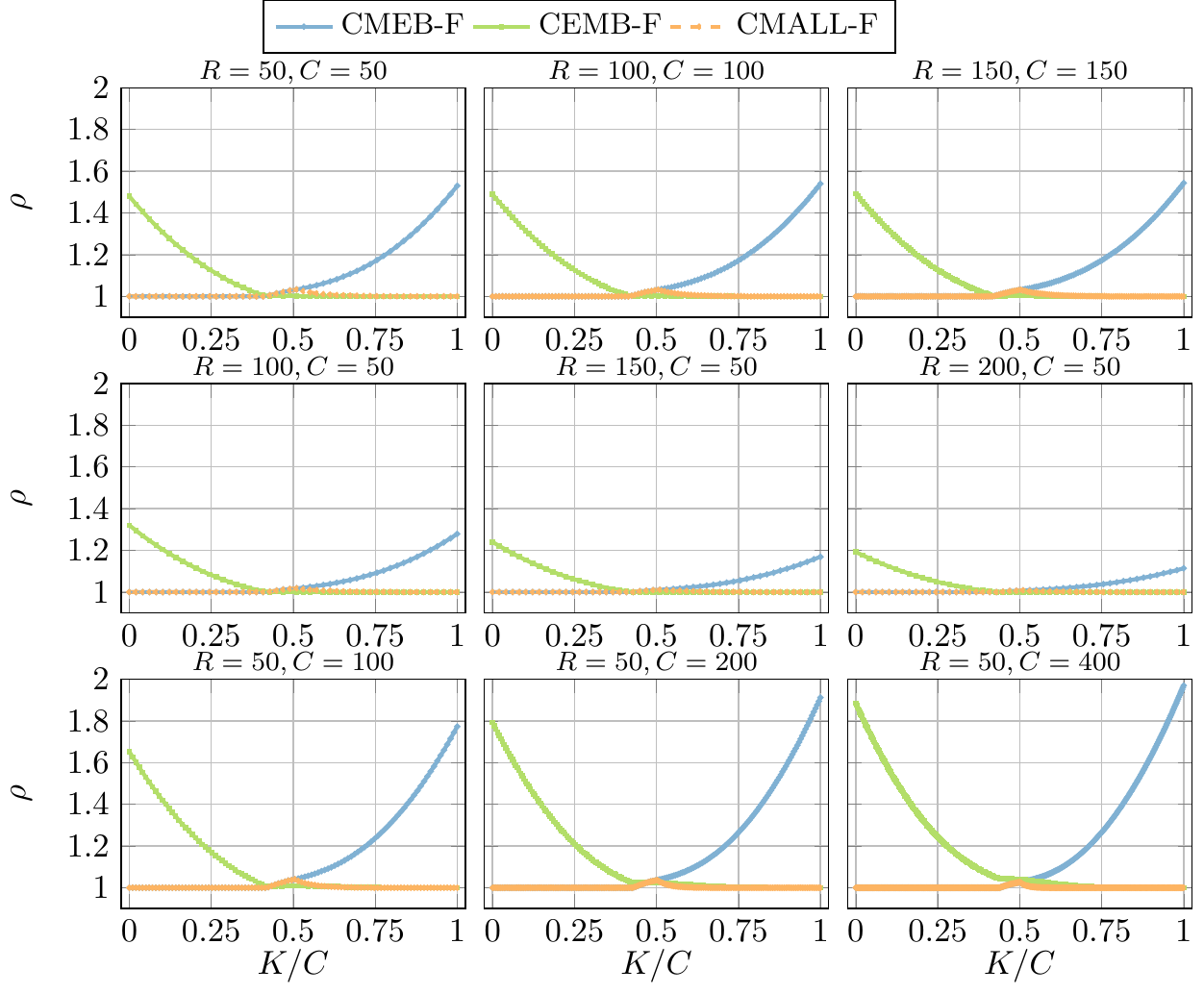}
    \caption{All the algorithms with the full-grid scenario.}
    \label{fig:full-three}
\end{figure}

\algcmallf performs very well since its DP \cmall is always very close to \copt.
On the other hand, as expected, \algcembf and \algcmebf are heavily affected by the value of $K$.
In fact, for small values of $K$, \algcembf reports bad results, while when $K$ increases the ratio $\rho$ tends to $1$.

We note that the performance of \algcembf is almost a reflection, in the vertical line $\frac{K}{C} = 0.5$, of the performance of \algcmebf. In particular, in Figure~\ref{fig:full-three}, we can observe that \algcmallf and \algcmebf perform similarly when $0 \le \frac{K}{C} \le 0.5$ (the two lines, i.e., the orange and the blue ones, almost coincide), and \algcmallf and \algcembf perform similarly when $\frac{K}{C} \le 0.5 \le 1$ (orange and green lines).

We also note that, when $R=C$ (Figure~\ref{fig:full-three}, first row), the worst cases of \algcembf and \algcmebf, i.e., the worst $\rho$ values exhibited, almost coincide; when $R>C$ (Figure~\ref{fig:full-three}, second row), the worst case of \algcmebf is slightly better than that of \algcembf; and when $R<C$ (Figure~\ref{fig:full-three}, third row) the worst case of \algcmebf is worse than that of \algcembf.

It is interesting to note that the performance of each algorithm is better when $R > C$, than when $R<C$.

\begin{figure}[ht]
    \centering
    \includegraphics[scale=0.9]{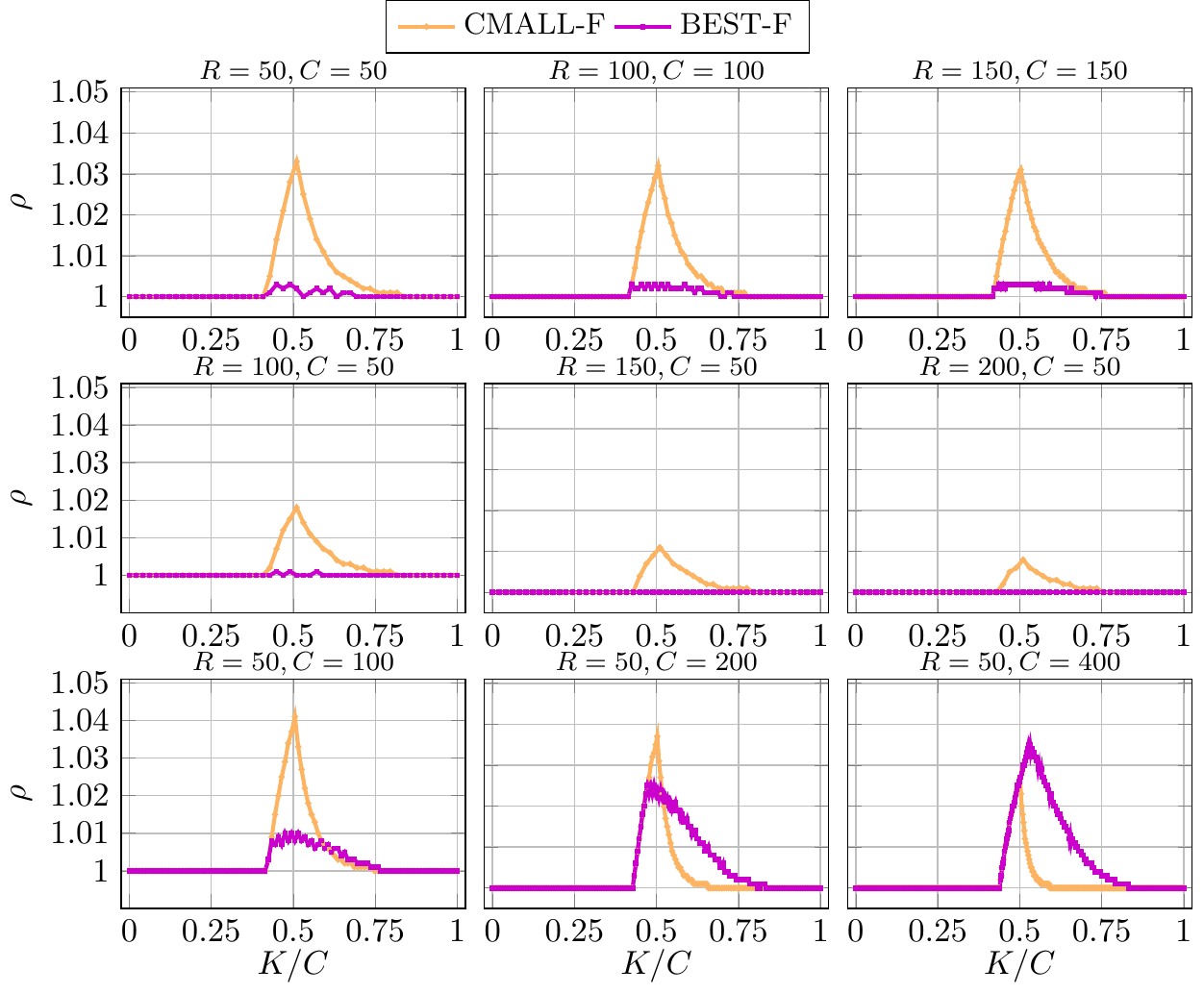}
    \caption{\algcmallf vs \algbestf with the full-grid scenario.}
    \label{fig:full-two}
\end{figure}

The almost symmetric performance, across $K/C = 0.5$, exhibited by \algcembf and \algcmebf motivates the design of a hybrid algorithm, \algbestf, that simply returns the best result among the two.
In Figure~\ref{fig:full-two} we compare \algcmallf against the best among \algcembf and \algcmebf (highlighted as \algbestf).
Since the experimental results are very close to the optimum ratio, which is $1$, here we reduce the scale along the $y$-axis.
In general, \algbestf almost always outperforms \algcmallf, and it is very close to \algoptf.
For instance, when $R=150$ and $C=50$, \algbestf is constantly $1$, and so we can state that when $R \gg C$, \algbestf is comparable to \algoptf.
With the scaled $y$-axis in Figure~\ref{fig:full-two}, it becomes clear \algcmallf has the worst performance when $\frac{K}{C} \approx 0.5$. %
Nevertheless, under these circumstances, the ratio $\rho$ is very low.
It is worth noting that in the case of $C \gg R$, \algbestf is not better than \algcmallf (e.g., $R=50, C=400$).
In fact, when $C=2R$ \algbestf is preferable, but when $C=8R$ (or more) it seems that \algcmallf is more consistent.
However, in hypothetical real cases with a reasonable mix of rows and columns, \algbestf basically performs like the optimal \algoptf, albeit more efficiently, since \algbestf is a constant time algorithm.
Finally, it can be seen that \algcmallf is always below the guaranteed threshold $\sqrt{2}$ on each plot.

\begin{figure}[ht]
    \centering
    \includegraphics[scale=0.9]{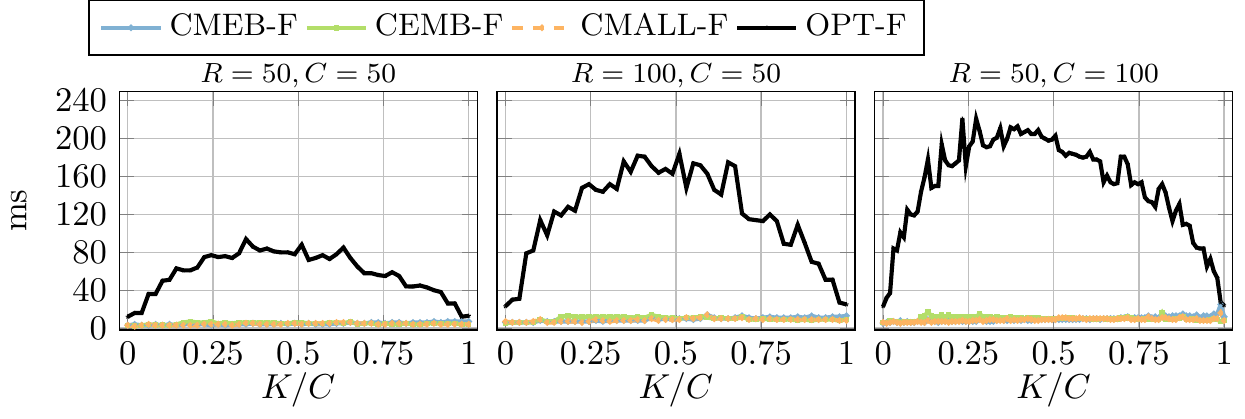}
    \caption{Running time of all the algorithms with the full-grid scenario.}
    \label{fig:full-time}
\end{figure}

\medskip\noindent\textit{Running Times}

In Figure~\ref{fig:full-time} we report the experimental running time (in milliseconds) of all the algorithms with the full-grid scenario.
In particular, the constant-time algorithms like \algcmallf take, on average, $2$--$6 \unit{ms}$, and therefore the graphs of their running times almost coincide.
On the other hand, the optimal logarithmic-time algorithm \algoptf takes much more time -- in the order of tenths of a second (in the worst case, $200 \unit{ms}$).
The experimental running time of the \algoptf algorithm is low when the grid is either almost Manhattan or almost Euclidean, i.e., when $\frac{K}{C} \to 0$ or $\frac{K}{C} \to 1$, respectively.

In the following we will show that the behavior of \algoptf is consistent with the claimed time complexity $\mathcal{O}(\log K)$.
To better understand this behavior we discuss the following example (for almost Manhattan grids): 
When $K$ is small with respect to $C$, say $\frac{K}{C}=0.25$ (see Figure~\ref{fig:full-time}), Algorithm \algoptf evaluates the column number, in the interval $[\hk, K]$, with the least cost (Line~\ref{code:algoptf:init-case1} of Algorithm~\ref{alg:algoptf}).
Specifically, it performs an efficient binary search in the interval $[\hk, K]$.
If we assume that $\frac{K}{C}=0.25$ or $\frac{K}{C}=0.5$ (Line~\ref{code:algoptf:init-case1} of Algorithm~\ref{alg:algoptf}), then the interval $[\hk, K]$ has width $\frac{K}{2}$.
Since when $C$ is fixed, varying $\frac{K}{C}$ from $0.25$ to $0.5$, the interval width $\frac{K}{2}$ increases from $\frac{C}{8}$ to $\frac{C}{4}$, the number of steps required by the \algoptf algorithm increase accordingly to the logarithm of $\hk$.

When $K$ is large with respect to $C$, i.e., the grid is almost all Euclidean, say $\frac{K}{C}=0.75$ (see Line~\ref{code:algoptf:init-case2} of Algorithm~\ref{alg:algoptf}), the \algoptf algorithm performs an efficient binary search in the interval $[\hk, \hc]$.
Recall that $\hc = \frac{C}{2}$.
If we assume that $\frac{K}{C}=0.75$, then $K=\frac{3}{4}C$ and $\hk=\frac{K}{2}=\frac{3}{8}C$, and therefore the interval $[\hk, \hc]$ has width $\frac{C}{8}$.
The time complexity is therefore logarithmic in $\frac{C}{8}$ as it was in the case $\frac{K}{C}=0.25$. 
Thus, the time complexity is comparable with that for $\frac{K}{C}=0.25$, as reported in Figure~\ref{fig:full-time}.

It is very important to note that, the plots in Figure~\ref{fig:full-time} also include the pre-processing time required by the Algorithm \algoptf; this pre-processing has time complexity $\mathcal{O}(RK + K)$ and hence increases with $R$ 
and $K \le C$.
In Figure~\ref{fig:full-time}, the pre-processing phase impacts more when $C=100$ because its time complexity depends on $K$, which can be $C$ in the worst case scenario.

\subsubsection{Partial-Grid Scenario}\label{subsubsec:partial-synthetic}

As before, we will first analyze our empirical results of the performance of the algorithms with respect to the delivery costs, and then with respect to the running times.

\medskip\noindent\textit{Delivery Costs}

Figure~\ref{fig:partial-A} compares the algorithms when solving \prob with the partial-grid scenario.
In particular, given the bad behavior of Algorithm \algcembf for large values of $K$ in the full-grid scenario (see Figure~\ref{fig:full-three}), for its partial-grid version, i.e., \algcembp, we only report the results using small values of $K$. For a similar reason, we report the results for Algorithm \algcmebp only for large values of $K$.
Moreover, recall that the best DP returned by \algcembp and \algcmebp is the optimum.

\begin{figure}[ht]
    \centering
    \includegraphics[scale=0.9]{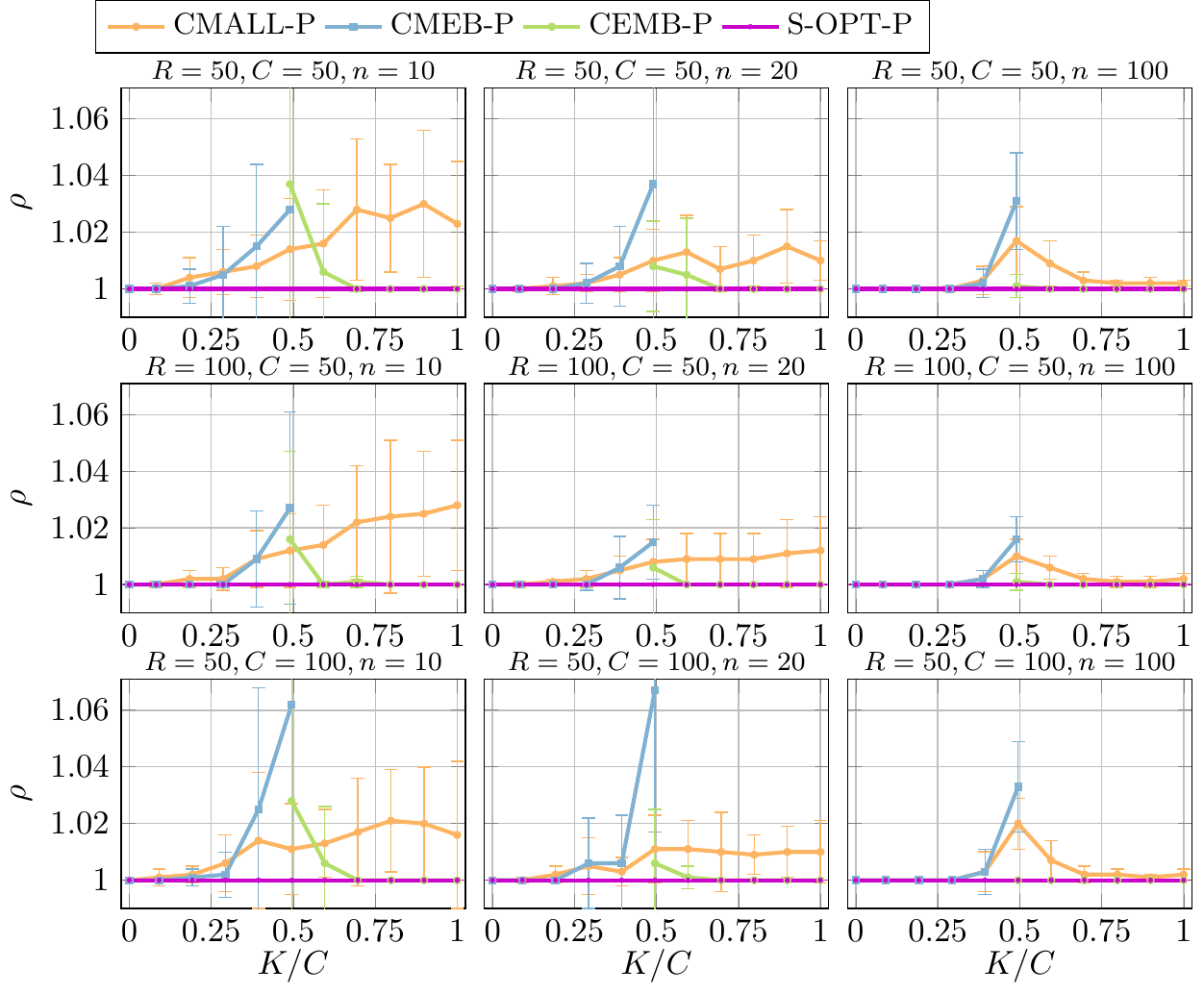}
    \caption{All the algorithms when varying $n$ with the partial-grid scenario.}
    \label{fig:partial-A}
\end{figure}

As before, in Figure~\ref{fig:partial-A} we present our empirical results in three groups, -- $R=C$, $R>C$, and $R<C$.
As expected, both the \algcembp and \algcmebp algorithms have different behavior, and one is more suitable than the other depending on the particular value of $K$.
When $n$ is small, all the algorithms are variable and the standard deviation is large.
Much more stable results can be observed when $n$ increases.
This is due to the fact that for larger values of $n$ the subset $H \subset G$ is ``denser'', and all the algorithms are less affected by the randomness. 
In general, the algorithms seem to perform better when $R > C$.
Of particular note is the behavior exhibited by Algorithm \algsuboptp -- it performs almost as well as Algorithm \algoptp! We discuss this further in Section~\ref{subsubsec:noteworthy} below.

\begin{figure}[ht]
    \centering
    \includegraphics[scale=0.9]{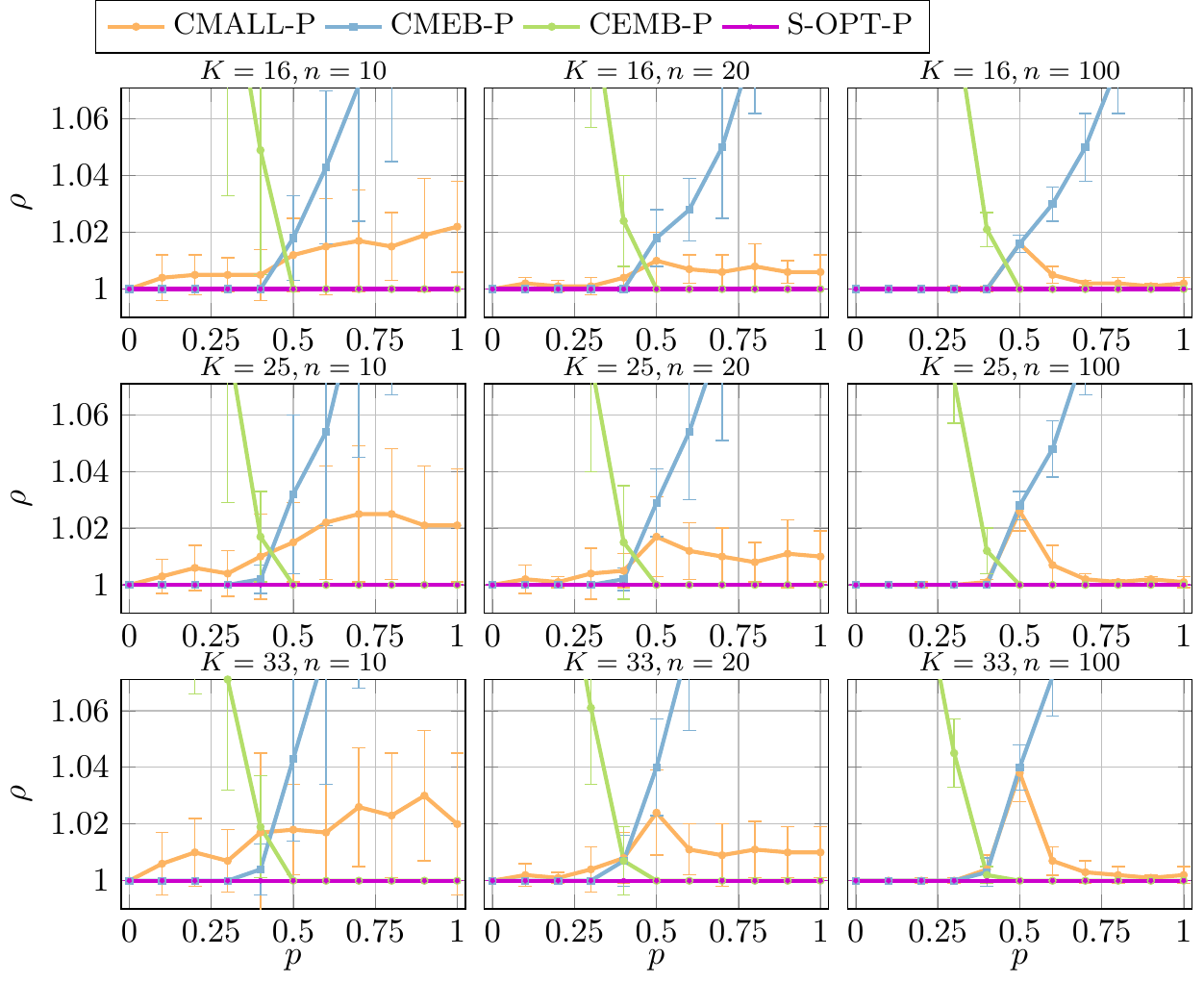}
    \caption{All the algorithms when fixing $R=50$ and $C=50$, and varying $p$ with the partial-grid scenario.}
    \label{fig:partial-B}
\end{figure}

In Figure~\ref{fig:partial-B} we present our empirical results for all the algorithms with the partial-grid scenario when varying the parameter $0 \le p \le 1$ such that $n_E=n \cdot p$ and $n_M=n \cdot (p-1)$.
For simplicity, in Figure~\ref{fig:partial-B} we only depict the squared layout, in which on each row we change the border position.
Not surprisingly, when $n$ is low the variability (standard deviation) is high, and algorithms like \algcmallp return a good ratio $\rho$ (average).
The other two algorithms, i.e., \algcembp and \algcmebp, have the usual symmetric behavior.
However, in this setting, such behavior is much more evident since we vary the number of deliveries to distribute in the two sub-areas.
So, for small values of $p$, \algcmebp is \textit{almost} optimum, and for large values of $p$ \algcembp is almost optimum.
Instead, when $p \approx 0.5$, the two algorithms lose their efficacy.
This can be seen for the three values of $K$ shown in Figure~\ref{fig:partial-B}.
Once again, surprising, but expected based on the results presented in Figure~\ref{fig:partial-A}, \algsuboptp performs almost as well as the optimum algorithm.

\begin{figure}[ht]
    \centering
    \includegraphics[scale=0.9]{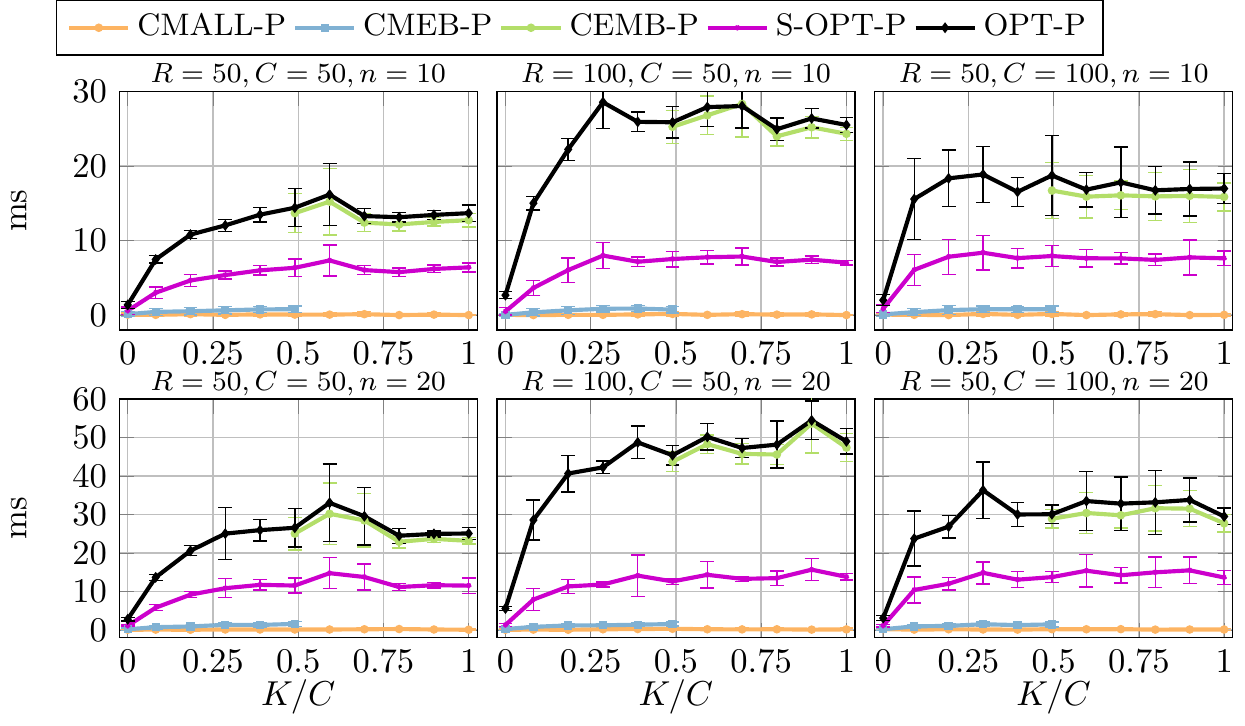}
    \caption{Running time of all the algorithms with the partial-grid scenario.}
    \label{fig:partial-A-time}
\end{figure}

\medskip\noindent\textit{Running Times}

Finally, in Figure~\ref{fig:partial-A-time} we report the experimental running time (in milliseconds) of all the algorithms with the partial-grid scenario.
Notice that the plots have a different scale in the $y$-axis, i.e., the ones in the first row have $30 \unit{ms}$, while the ones in the second row have $60 \unit{ms}$.
This has been done to emphasize the linear dependency on $n$: when $n$ doubles, the time performance on the $y$-axis doubles too, but since the scale of $y$ is halved, the behavior appears to be the same.
As expected, the fastest algorithm is \algcmallp whose time complexity is $\mathcal{O}(n)$, followed by \algcmebp with a time complexity of $\mathcal{O}(nR)$ (see Table~\ref{tab:comp_algs_partial}).
In Figure~\ref{fig:partial-A-time}, it is difficult to appreciate the actual difference between \algcmallp and \algcmebp, but, for instance, when $\frac{K}{C}=0.5$, on average, \algcmallp takes $0.061 \unit{ms}$, while \algcmebp takes $0.786 \unit{ms}$ (i.e., $10 \times$ more than \algcmallp).

The running time of the \algcmebp is ``linear'' since its cost does not depend on the value of the border $K$.
The sub-optimal algorithm \algsuboptp experimentally performs very well (theoretically it takes $\mathcal{O}(n \log R \log K)$ in time), in the order of $5$--$15 \unit{ms}$, with respect to the number of deliveries $n$.
The \algcembp algorithm, which has a time complexity of $\mathcal{O}(nR \log K)$, depends on the binary search function invoked $R$ times.
The binary search itself depends on the size of the border $K$, and therefore its running time increases (by a factor of $\log K$) when the \EM tends to become more ``Euclidean''.

In the full-grid scenario (Figure~\ref{fig:full-time}), the algorithm that exploits the binary search (i.e., \algoptf) has a parabolic trend, justified by the fact that when $K$ is either small or large with respect to $C$, Algorithm~\ref{alg:algoptf} executes differently.
In this case (partial-grid scenario), though, we do not have a similar behavior, and the time complexity of Algorithm \algcembp is simply affected by the factor $\log K$.
Nevertheless, it is still possible to observe a light drop for \algcembp when $\frac{K}{C} \approx 0.75$.
Obviously, the optimal algorithm \algoptp, whose time complexity is $\mathcal{O}(nR \log K)$, takes the most time (as seen in Figure~\ref{fig:partial-A-time}) since it runs both \algcembp and \algcmebp, and returns the best point among the ones outputted by both.

\subsubsection{Noteworthy Observations}\label{subsubsec:noteworthy}

Based on the theoretical and empirical analyses above, we make the following noteworthy observations about our algorithms for solving \prob in the partial-grid scenario.

In general, Algorithm \algcmallp has the best trade-off with respect to obtained performance and time complexity.
In fact, when varying the number $n$ of deliveries, on average, its ratio $\rho$ is always lower than $1.03$. Considering its time complexity -- linear with respect to the number of deliveries -- \algcmallp is a very good compromise.
As noted above, both the \algcembp and \algcembp algorithms have different behavior, and one is more suitable than the other depending on the particular value of $K$.
However, as proven in Section~\ref{sec:full:properties}, their combination, by returning the best among the two, ensures finding the optimal DP.
So, despite their swinging behavior and relatively high complexity, they can be used to find the optimal DP.

We noted above the surprising, almost-optimal behavior of the sub-optimal algorithm \algsuboptp:
On the $33$ random instances, for each combination of rows, columns, border, and number of deliveries, the returned DP was also the optimal one.
For this reason, on each plot of Figure~\ref{fig:partial-A}, \algsuboptp exhibits a completely flat line.
However, on an extensive ``brute-force'' campaign with $10,000+$ random instances, we found $17$ instances where the returned point \csopt was different from the optimal \copt.
Nevertheless, in these $17$ instances the worst ratio $\rho$ was $1.005$.
Of course, we have found counterexamples that show that \algsuboptp is not optimum, but it performs extremely well with a time complexity much lower than that of \algoptp.

\subsection{Results with Quasi-real Data}\label{sec:real}

In this section, we test our algorithms in the quasi-real case.
Specifically, we evaluate the performance of our partial-grid scenario algorithms applied on a real-world map.
Clearly, in this case, we still generate random instances of customers.

\subsubsection{Settings}

In this section, we first discuss how the delivery maps are extracted from real cities, and then which commercial drones are capable of delivering products in these areas.

\medskip\noindent\textit{Map Extraction}

In order to evaluate our algorithms in a real-world case, we initially needed to extract the \EMs from real cities.
Figure~\ref{img:examples} illustrates how to roughly perform this extraction task by modeling Chicago, New York, and Miami, in the top, middle, and bottom, respectively.

\begin{figure}[!ht]
    \centering
    \subfloat[Chicago.]{%
        \includegraphics[width=6.0cm]{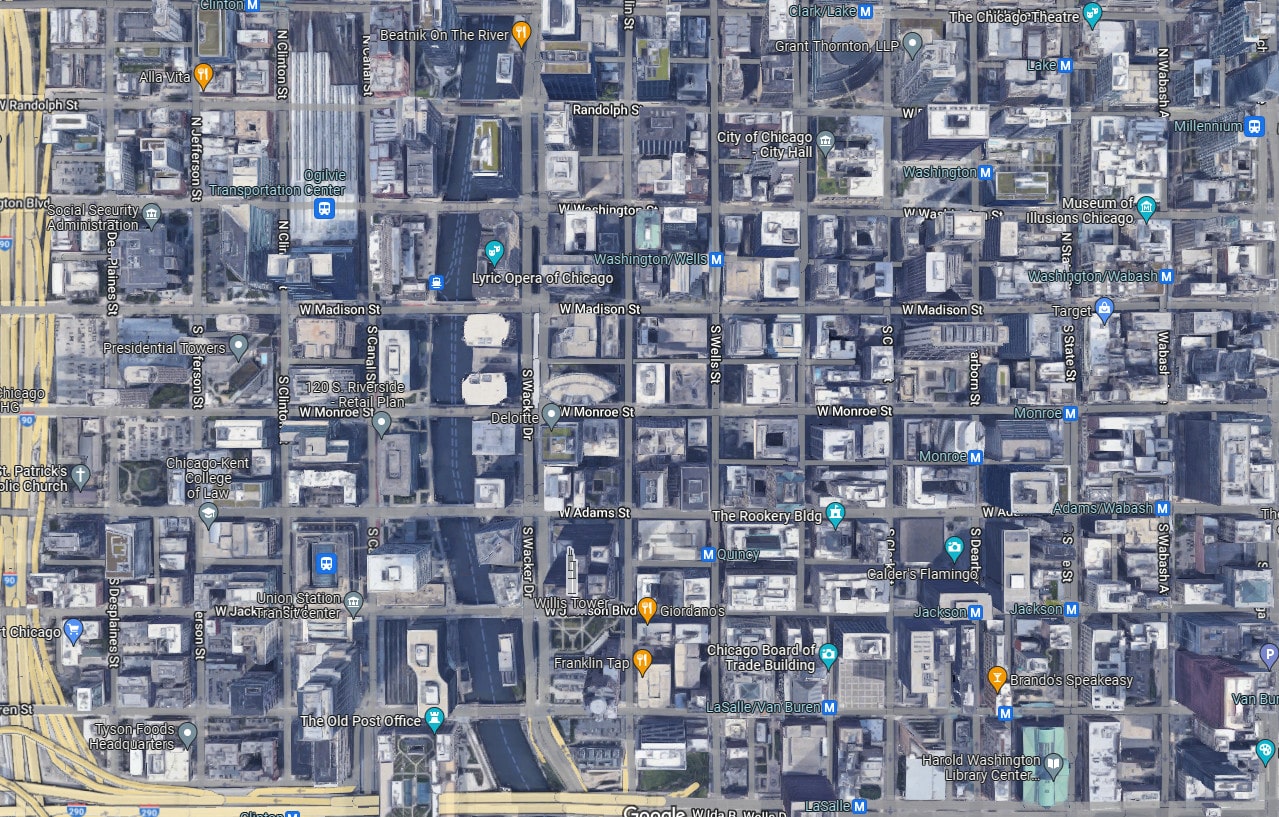}
        \label{img:chicago-original}
    }
    \subfloat[The corresponding \EM $G_1 = (8, 14, 6)$.]{%
        \includegraphics[width=6.0cm]{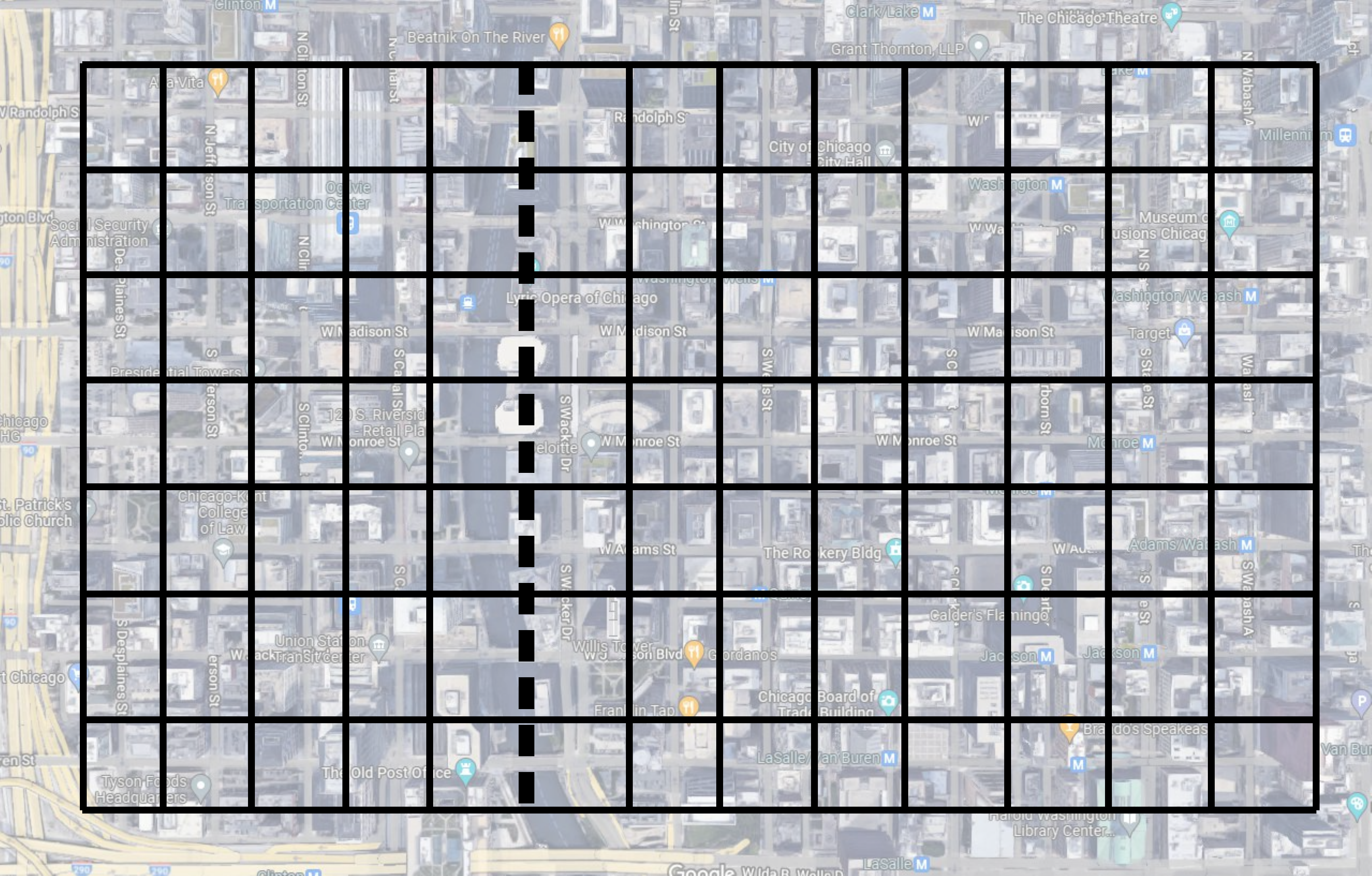}
        \label{img:chicago-grid}
    }
    \hfill
    \subfloat[New York.]{%
        \includegraphics[width=6.0cm]{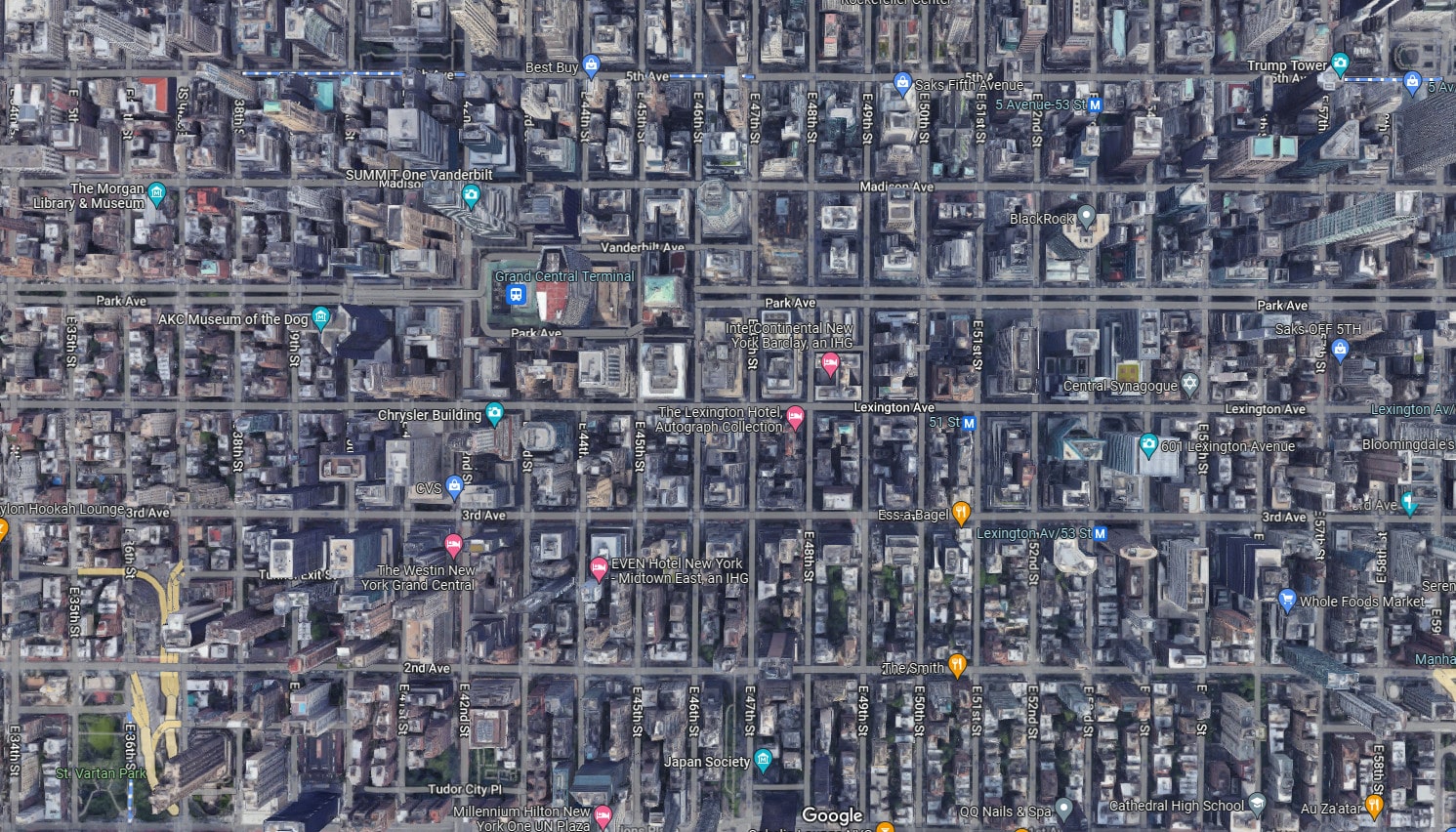}
        \label{img:newyork-original}
    }
    \subfloat[The corresponding \EM $G_2 = (7, 24, 1)$.]{%
        \includegraphics[width=6.0cm]{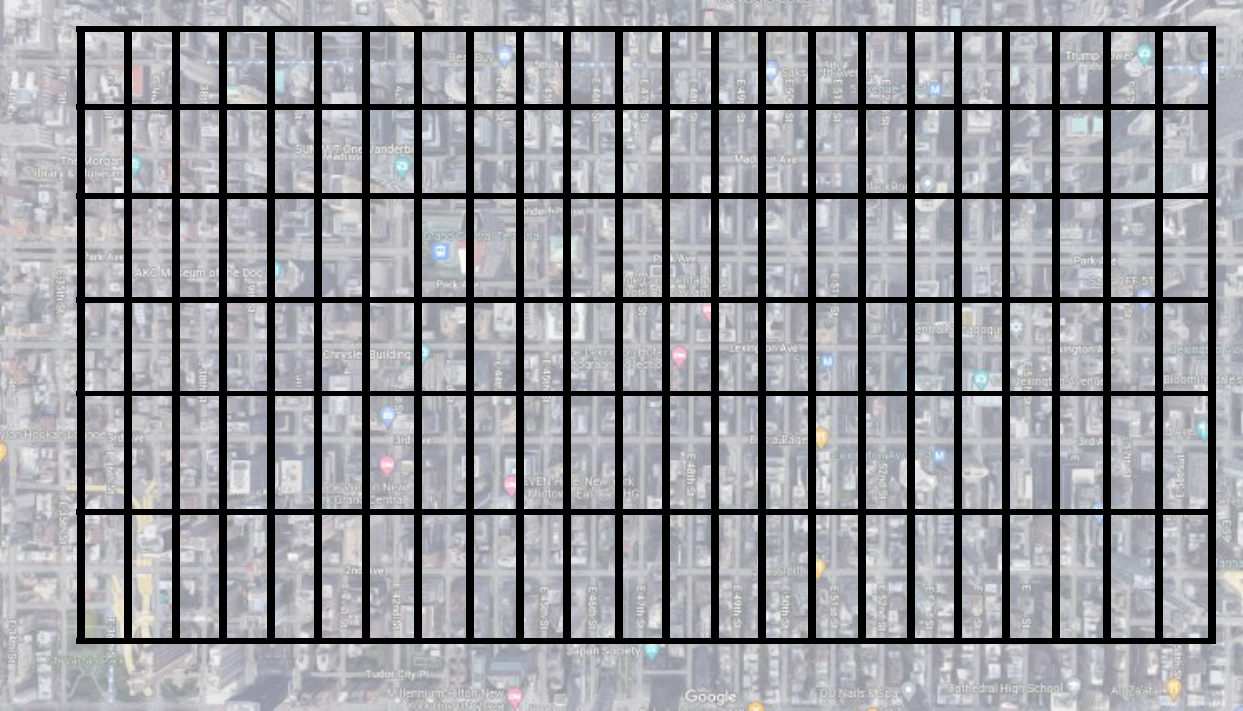}
        \label{img:newyork-grid}
    }
    \hfill
    \subfloat[Miami.]{%
        \includegraphics[width=6.0cm]{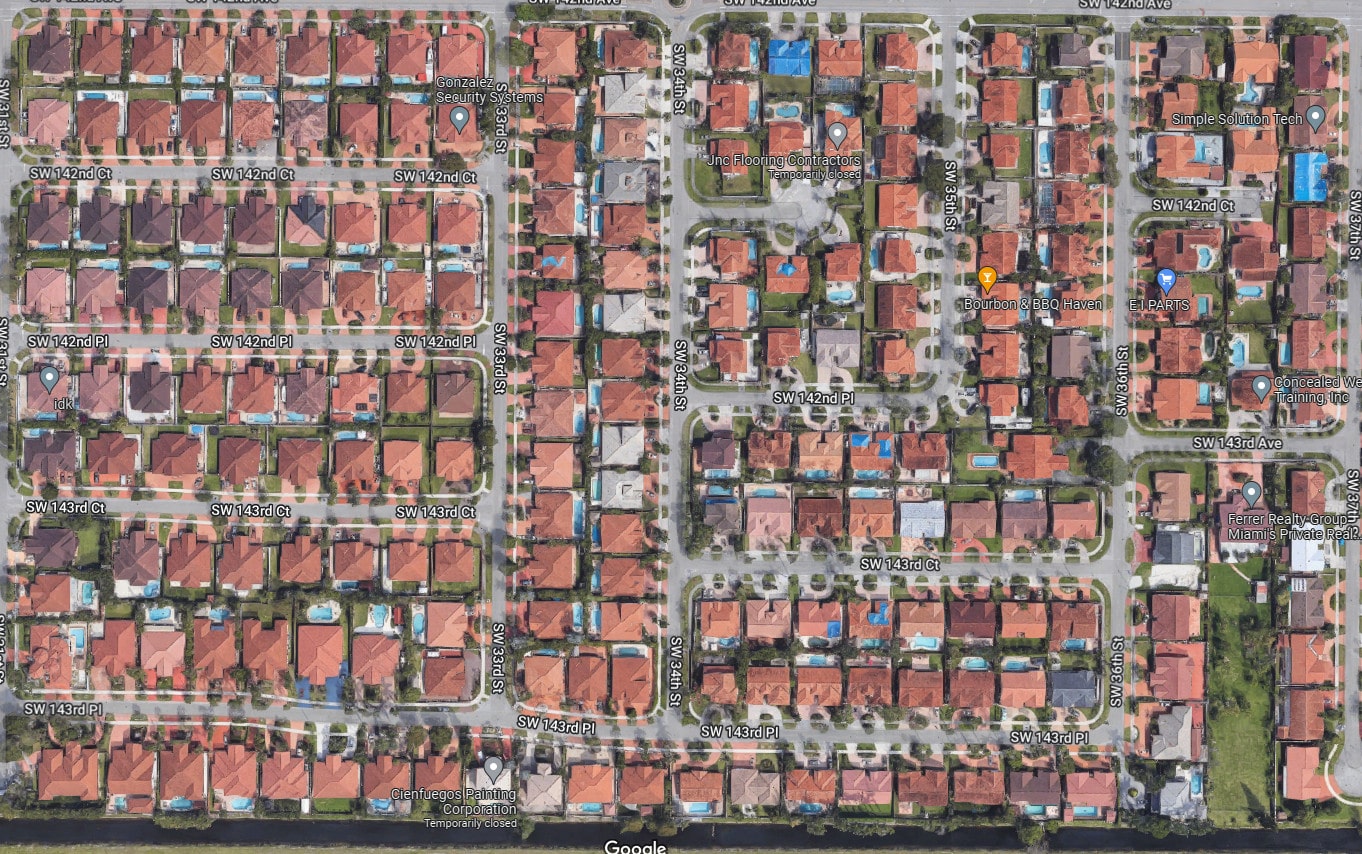}
        \label{img:miami-original}
    }
    \subfloat[The corresponding \EM $G_3 = (9, 20, 20)$.]{%
        \includegraphics[width=6.0cm]{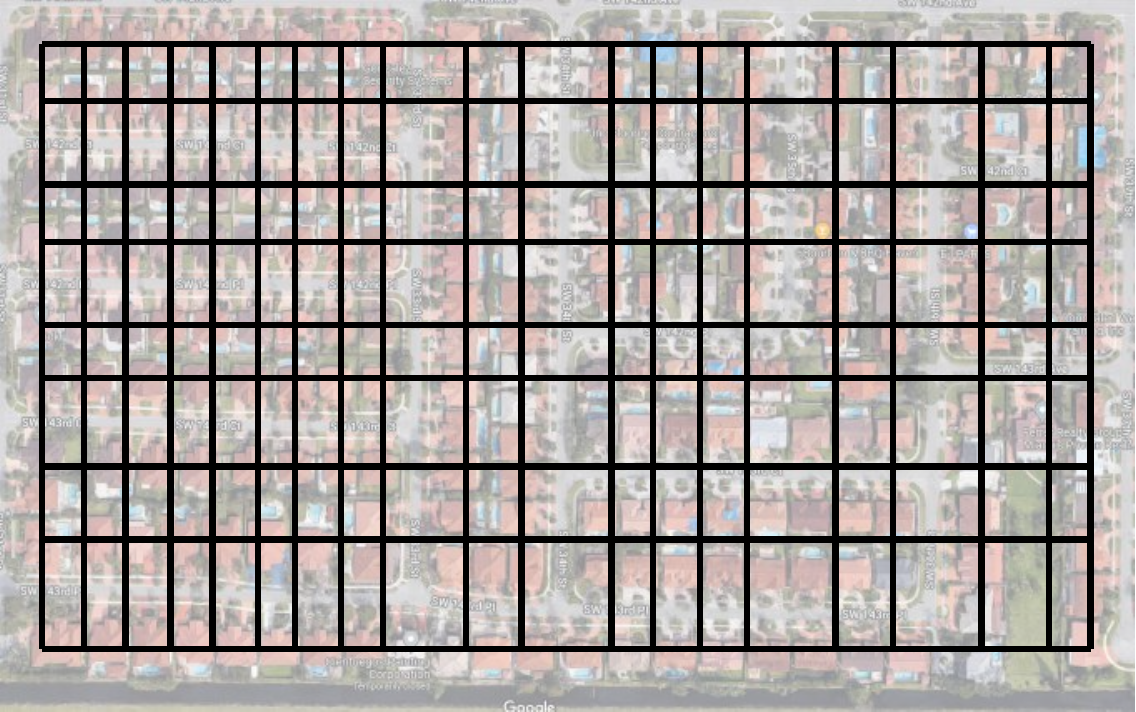}
        \label{img:miami-grid}
    }
    \caption{The original cities and the corresponding extracted \EMs.}
    \label{img:examples}
\end{figure}

In Figure~\ref{img:chicago-original}, there is the top view of a portion of the city of Chicago, while in Figure~\ref{img:chicago-grid} the corresponding \EM, i.e., $G_1 = (8, 14, 6)$.
The average length of a ``block'' is $\approx 120 \unit{m}$, and so the distance among two adjacent vertices can be set.
So, this map $G_1$ has a size of $840 \unit{m} \times 1560 \unit{m}$.
Roughly, the river from the top to the bottom of Figure~\ref{img:chicago-original} splits the portion of Chicago in two contiguous areas, i.e., the left one with relatively low buildings (although some tall buildings are present), and the right one with lots of skyscrapers (although some low buildings are present).
The dashed line in Figure~\ref{img:chicago-grid} delimits the border between the two areas.

In Figure~\ref{img:newyork-original}, there is the top view of a portion of the city of New York, while in Figure~\ref{img:newyork-grid} the corresponding \EM, i.e., $G_2 = (7, 24, 1)$.
This means that the extracted \EM is a full Manhattan grid.
In this case, differently from Chicago, we can observe that the length of a block on ``streets'' (horizontal) is $\approx 80 \unit{m}$, while the length of a block on ``avenues'' (vertical) is $\approx 150 \unit{m}$.
So, this map $G_2$ has a size of $660 \unit{m} \times 2530 \unit{m}$.
Since in our model both the lengths must be equal, we average the two and set a value of $110 \unit{m}$.

Finally, in Figure~\ref{img:miami-original}, there is the top view of a portion of the city of Miami, while in Figure~\ref{img:miami-grid} the corresponding \EM, i.e., $G_3 = (9, 20, 20)$.
This means that the extracted \EM is a full Euclidean grid.
Even the pure Euclidean maps are hard to find since roads and houses can have different sizes, lengths, and so on.
The portion of Miami city reported in Figure~\ref{img:miami-original} contains some irregularities, especially in the right portion.
Here, the average distance among two adjacent houses is in the order of $\approx 25 \unit{m}$.
So, this map $G_3$ has a size of $200 \unit{m} \times 475 \unit{m}$.

\medskip\noindent\textit{Drone Selection}

Previously, we have seen how to extract an \EM from a real city.
Therefore, the next step is the drone selection.
There are a plethora of off-the-shelf drones that are available at the moment, but only a very few have the required characteristics for performing deliveries.
One of the most common drones able to do deliveries with some autonomy is the DJI Matrice 300 RTK (briefly Matrice).
According to~\cite{matrice300}, this drone can carry up to $2.7 \unit{kg}$, can fly up to $8 \unit{km}$ away from the remote controller, can fly for $30 \unit{min}$ (with maximum payload), and can fly up to $17 \unit{m/s} \approx 60 \unit{km/h}$ in P-mode (i.e., the Positioning mode recommended in autonomous flights).
Clearly, the temporal bound of $30 \unit{min}$ also depends on the drone's speed and the payload as well as the current weather conditions.
Nevertheless, we still continue to use these numbers as a reference.
So, the Matrice can fly a distance of approximately $30 \unit{km}$.
Recalling that the previous three grids $G_1$, $G_2$, and $G_3$ have an area of approximately $1.2 \unit{km^2}$, $1.5 \unit{km^2}$, and $0.1 \unit{km^2}$, respectively, we believe that the Matrice can operate in these neighborhoods of the cities.

\subsubsection{Results}

Finally, Figure~\ref{fig:real} compares the algorithms in the partial-grid scenario in a quasi-real case.
In particular, in the $x$-axis we report the number of deliveries $n$, while in the $y$-axis we report the traveled drone's distance in $\unit{km}$.
As aforementioned, the selected Matrice drone can fly for approximately $30 \unit{km}$, and for this reason we put a solid red line on that threshold in each plot of Figure~\ref{fig:real}.
Moreover, recall that $G_1$, $G_2$, and $G_3$ correspond to the map of Chicago (mixed \EM), New York (full Manhattan \EM), and Miami (full Euclidean \EM), respectively.

\begin{figure}[ht]
    \centering
    \includegraphics[scale=0.9]{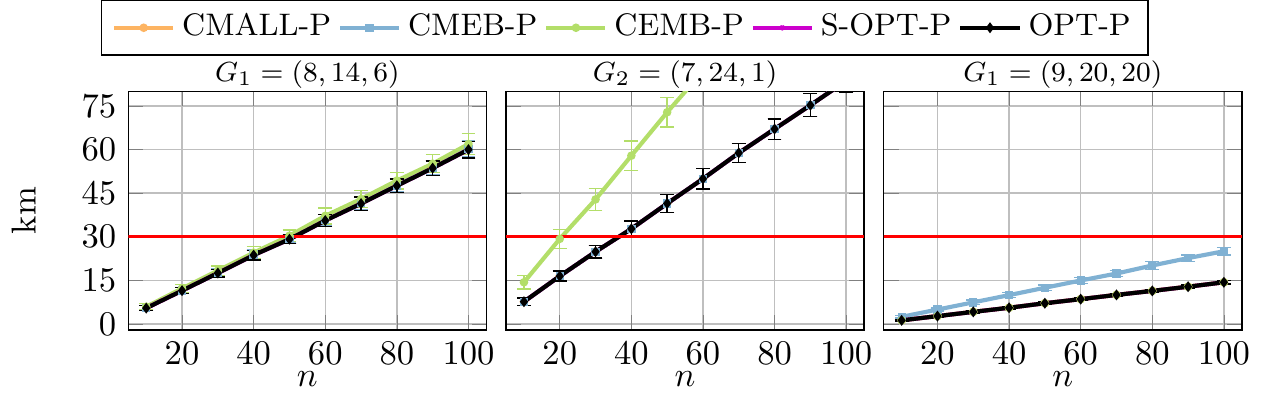}
    \caption{Traveled drone's distance when comparing all the algorithms in the quasi-real case. The \EMs are the ones extracted from the maps in Figure~\ref{img:examples}.}
    \label{fig:real}
\end{figure}

Our first observation is that the differences among the algorithms, in terms of distance ($\unit{km}$) traveled, is small, with a couple of exceptions.
In fact, as we discussed in the synthetic case, it is not profitable at all to perform the \algcembp algorithm when the \EM is a full Manhattan grid (see $G_2$ in Figure~\ref{fig:real}), or similarly, to perform the \algcmebp algorithm when the \EM is a full Euclidean grid (see $G_3$).
In the general mixed case (see $G_1$), though, all the algorithms perform decently with respect to the optimal one.

Taking into account the Chicago map, i.e., $G_1$, we can immediately observe that the adopted Matrice drone can potentially accomplish the whole mission of up to $n=50$ deliveries with a single battery charge, regardless of the algorithm.
Specifically, on average, both \algoptp and \algsuboptp require $29.181 \unit{km}$ of travel, followed by $29.297 \unit{km}$ for \algcmallp, $29.324 \unit{km}$ for \algcmebp, and finally $30.395 \unit{km}$ for \algcembp.
We note that the drone, if employing Algorithm \algcembp, may be forced to fly beyond the aforementioned and approximated threshold of $30 \unit{km}$.
However, here one can appreciate that the difference, on average, among the best and the worst performing algorithm is $\approx 1 \unit{km}$, which can be detrimental in some borderline cases.
Moreover, even neglecting the \algcembp algorithm, the above difference is still about $0.143 \unit{km}$.
Recalling that the average block length in $G_1$ is $0.12 \unit{km}$, this small difference could potentially result in a delivery not being completed.

In the $G_2$ map in New York City, which is a full Manhattan grid, we can safely do up to $n=35$ deliveries.
In this case, we cannot rely to the \algcembp algorithm, since it is not suitable for full Manhattan grids.
It is important to recall that, in $G_2$, \algcmallp returns, by definition of the Manhattan median, the optimal solution.

Finally, in the full Euclidean grid in Miami ($G_3$), due to the smaller distances in the example, the Matrice drone can safely perform even more than $n=100$ deliveries.
Even relying to the \algcmebp algorithm, we are still able to serve $100$ customers in the area.
It is important to recall that, in $G_3$, \algcembp returns the optimal solution.

\section{Conclusion}\label{sec:conclusion}
We considered a drone-based delivery system for the ``last-mile'' logistics of small parcels, medicines, or viral tests, in \EMs. 
The shortest path in an \EM concatenates Euclidean and  Manhattan distances.
We solved the \prob on \EMs whose goal is to minimize the sum of the distances between the locations to be served and the drone's DP.
Finding the most suitable DP has many implications that can impact the expected delivery time for customers, the energy consumption of drones, and in general, the broader environmental impacts, e.g., the quantity of CO\textsubscript{2} emissions, when relying on trucks.
We propose efficient algorithms to exactly solve the problem in \EMs in both the full-grid and the partial-grid scenarios under the assumption that the vertex distance is unitary.
We also run our algorithms on quasi-real cases by extracting the \EMs from real cities in the United States.

Although we attempted to run our algorithms on real cities, we are aware that our \EM model is too simplified to characterize any real-world scenario.
In future work, we intend to extend the introduced mixed-grid model to more general layouts (e.g., a rural area inside an urban area in Central Park, New York).
Nevertheless, we believe that our work could be the starting point to devising much more complex scenarios to better model real-world scenarios.
For example, we could use our technique of map extraction (see Section~\ref{sec:real}) to construct a grid, $G$, where the each individual grid square is labeled as ``Buildings'' or ``Park''. A contiguous group of squares labeled ``Buildings'' is effectively a Manhattan grid, and a contiguous group of squares labeled ``Park'' is a Euclidean grid. Then, we could sub-divide $G$ into multiple \EMs, and apply our algorithms suitably in each \EM.

Another interesting variant to study is the use of multiple drones that are responsible for delivering packages to different partitions of the customers.
In this case, we will need efficient clustering algorithms for minimizing suitable metrics.

\bibliographystyle{IEEEtran}
\bibliography{references}

\end{document}